\pdfoutput=1 

\RequirePackage{snapshot} 
\RequirePackage{sty/bootstrap}
\PrependInputPath{fig}

\ProvideOptions{notecolumn, blackqed, final, arxiv}

\WithOptions
\documentclass {IEEEtranPlus}

\ProvideIfOptions {ieee, conference, arxiv, compress, final, blind}
\ProvideIf[true]{draft}
\iffinal\draftfalse\fi

\ifdraft
\fi

\ifcompress
  \IfClassLoadedTF {IEEEtran} {

  }
\fi

\PassOptionsToPackage{fixmeauthors={ai,jr,bb,cd,rf,up}}{boilerplate}

\ifieee
    \usepackage[hyperref=false,bibtex]{boilerplate}
\else
  \ifarxiv
    \usepackage[biblatex]{boilerplate}
    
  \else
    \usepackage[bibtex]{boilerplate}
  \fi
\fi

\ifieee
  \DeclareDocumentCommand{\orcidlink}{m}{}
\fi

\usepackage{multicol}
\usepackage{tabularray}

\DefTblrTemplate{conthead-text}{none}{}
\NewTblrTheme{plain}{
  \SetTblrTemplate{conthead-text}{none}
}

\ifcompress\else
  \PassOptionsToPackage{noabbrev}{cleveref}
\fi
\usepackage[capitalize]{cleveref}
\crefname{equation}{}{}
\Crefname{equation}{Equation}{Equations}

\tryWithFiles

\NewDocumentCommand \arxivOrcid {m} {    \ifarxiv\orcidlink{#1}\fi,}

\title{      Identification Under the Semantic Effective Secrecy Constraint
}
\ifblind\else
\author{  \IEEEauthorblockN{    Abdalla~Ibrahim\arxivOrcid{0009-0008-7646-6276}
      \IEEEmembership{Graduate~Student~Member,~IEEE},
    Johannes~Rosenberger\arxivOrcid{0000-0003-2267-3794}
      \IEEEmembership{Graduate~Student~Member,~IEEE},
    Boulat~A.~Bash\arxivOrcid{0000-0002-1205-3906}
      \IEEEmembership{Member,~IEEE},
    Christian~Deppe\arxivOrcid{0000-0002-2265-4887}
      \IEEEmembership{Senior Member,~IEEE},
    Roberto~Ferrara\arxivOrcid{0000-0002-1991-3286}
    and
    Uzi~Pereg\arxivOrcid{0000-0002-3259-6094}
      \IEEEmembership{Member,~IEEE}
  }                                 \thanks{    Abdalla Ibrahim, Johannes Rosenberger and Roberto Ferrara are with the
    Technical University of Munich, Germany;
    TUM School of Computation, Information and Technology,
    Department of Computer Engineering,
    (e-mail: \{\tumail{abdalla.m.ibrahim},\tumail{johannes.rosenberger},\tumail{roberto.ferrara}\}@tum.de).
    Boulat Bash is with the
    Electrical and Computer Engineering Department, University of Arizona, USA
    (e-mail: \uref{mailto:boulat@arizona.edu}{boulat@arizona.edu}).
    Christian Deppe is with the Institute for Communications Technology,
    Technische Universit\"{a}t Braunschweig
    (e-mail: \uref{mailto:christian.deppe@tu-braunschweig.de}{christian.deppe@tu-braunschweig.de}).
    Uzi Pereg is with the Faculty of Electrical and
    Computer Engineering at the Technion -- Israel Institute of Technology and
    the Technion Hellen Diller Quantum Center, 3200003 Haifa, Israel
    (e-mail: \uref{mailto:uzipereg@technion.ac.il}{uzipereg@technion.ac.il}).  }  \thanks{    Parts of this paper have been presented at the
    2021 IEEE Information Theory Workshop (ITW) \cite{IbrahimFerraraDeppe2021IdentificationEffectiveSecrecy_conference}
    and at the 2023 IEEE International Symposium on Information Theory (ISIT) \cite{Rosenberger2023}.  }}

\makeatletter
\hypersetup{
  pdftitle = \string\@title,
  pdfauthor = {Abdalla Ibrahim, Johannes Rosenberger, Boulat A. Bash, Christian Deppe, Roberto Ferrara, Uzi Pereg},
  pdfkeywords = {},
  pdflang = {English},
}
\makeatother

\usepackage{tikz}
\ProvideIf[false]{qtikz}
\usetikzlibrary{math}
\usetikzlibrary{fpu}
\pgfkeys{/pgf/fpu}
\pgfkeys{/pgf/fpu/output format=fixed}
\tikzmath{
  function fold(\p, \q) { return \p*(1-\q) + \q*(1-\p); };
  function Hbin(\p) { return - \p * log2(\p) - (1 - \p)*log2(1-\p); };
  function Ibsc(\p, \q) { return Hbin(fold(\p, \q)) - Hbin(\q)); };
  function Ibec(\p, \q) {
    return Hbin(\p) - \q * Hbin(\p);
  };
  function IgaussC(\s) { return log2(1 + \s); };
  function IgaussR(\s) { return IgaussC(\s)/2; };
  function IzChannel(\p, \d) { return Hbin(\p * \d) - \p * Hbin(\d); };
  function zChannelPopt(\d) {
    real \g;
    let \g = (1-\d)^((1-\d)/\d);
    return \g/(1 + \d*\g);
  };
}
\pgfkeys{/pgf/fpu=false}

\providecommand \ifqtikz {\iftrue}

\usetikzlibrary{fpu} 
\usetikzlibrary{math} 
\usetikzlibrary{datavisualization} 
\usetikzlibrary{datavisualization.formats.functions}

\ifqtikz

\usetikzlibrary{math}
\usetikzlibrary{fpu}
\pgfkeys{/pgf/fpu}
\pgfkeys{/pgf/fpu/output format=fixed}
\tikzmath{
  function fold(\p, \q) { return \p*(1-\q) + \q*(1-\p); };
  function Hbin(\p) { return - \p * log2(\p) - (1 - \p)*log2(1-\p); };
  function Ibsc(\p, \q) { return Hbin(fold(\p, \q)) - Hbin(\q)); };
  function Ibec(\p, \q) {
    return Hbin(\p) - \q * Hbin(\p);
  };
  function IgaussC(\s) { return log2(1 + \s); };
  function IgaussR(\s) { return IgaussC(\s)/2; };
  function IzChannel(\p, \d) { return Hbin(\p * \d) - \p * Hbin(\d); };
  function zChannelPopt(\d) {
    real \g;
    let \g = (1-\d)^((1-\d)/\d);
    return \g/(1 + \d*\g);
  };
}
\pgfkeys{/pgf/fpu=false}

\tikzset{
  font=\sffamily,
  rateBound/identification/.style = {color=identification},
  rateBound/transmission/.style = {dotted,color=transmission}
}

\fi

\newcommand \tikzMutInfsRevDegraded [3][]{

\providecolor{identification}{named}{blue}
\providecolor{transmission}{named}{red}

\tikzmath{
 real \ticklen, \perasure, \pcrossover, \capBSC;
 \ticklen = 0.04; 
 \perasure = #3;
 \pcrossover = #2;
 function IbscBec(\a, \p, \e) { return (1-\e) * ( Hbin(fold(\a,\p)) - Hbin(\p) ); };
}

\begin{tikzpicture}[auto,thick,#1]
  \datavisualization [school book axes={unit=0.25},
    visualize as smooth line/.list={Ixy,Ixz,Iuy,Iuz,Ix1y,It},
    x axis={ticks={step=0.5},label={$p_{X_2}$}},
    y axis={ticks={step=0.5},label={bits/channel use}},
    legend={above right of={x=0.9,y=0.63}},
    style sheet=strong colors,
    style sheet=vary dashing,
    Ixy = {label in legend={text=$I(X;Y)$ [secret ID]}},
    Ixz = {label in legend={text=$I(X;Z)$}},
    Iuy = {label in legend={text=$I(U;Y)$ [ESID]}},
    Iuz = {label in legend={text=$I(U;Z)$}},
    Ix1y = {label in legend={text=$I(X_1;Y)$ [ES transmission]}},
    It = {label in legend={text=$I(U;Y) - I(U; Z)$}},
    data/format=function,
  ]
  data [set=Ixy] {
    var x : interval [0.001:1];
    func y = Ibec(0.5, \perasure) + Ibec(\value x, \perasure);
  }
  data [set=Ixz] {
    var x : interval [0.001:0.999];
    func y = Hbin(\value x);
  }
  data [set=Iuy] {
    var x : interval [0.001:0.999];
    func y = Ibec(0.5, \perasure) + IbscBec(\value x, \pcrossover, \perasure);
  }
  data [set=Iuz] {
    var x : interval [0.001:0.999];
    func y = Ibsc(\value x, \pcrossover);
  }
  data [set=Ix1y] {
    var x : interval [0.001:0.999];
    func y = Ibec(0.5, \perasure);
  }
  data [set=It] {
    var x : interval [0.001:0.999];
    func y = Ibec(0.5, \perasure) + IbscBec(\value x, \pcrossover, \perasure) - Ibsc(\value x, \pcrossover);
  };
\end{tikzpicture}
}

\ifqtikz
\tikzMutInfsRevDegraded{1/8}{0.6866}
\fi

\providecommand \ifqtikz {\iftrue}

\usetikzlibrary{positioning,fit}

\ifqtikz
  \tikzset{
    box/.style={draw, minimum height = .8cm, minimum width = 1cm, inner sep=4pt}, 
    pfeil/.style={->, >=latex}
  }
  \providecommand{\BEC}{\mathrm{BEC}}
  \providecommand{\BSC}{\mathrm{BSC}}
\fi

\NewDocumentCommand \exampleChannelESID { O{} O{1.5} }{

\begin{tikzpicture}[auto,#1]

  \node[box,minimum height=2*#2\baselineskip+2em] (Pu) {$P_{U}$};
  \coordinate[yshift = #2\baselineskip] (Pu1) at (Pu.east);
  \coordinate[yshift = -#2\baselineskip] (Pu2) at (Pu.east);
  \path
    (Pu1)
    ++(4,0)
    node[box](BEC1)  {$\BEC(\epsilon)$}

    (Pu2)
    ++(1.6,0)
    node[box] (Pxu2) {$\BSC(\delta)$}
    ++(2.4,0)
    node[box] (BEC2) {$\BEC(\epsilon)$}

    (Pu)
   ++(6.5,0)
    node[box,minimum height=2*#2\baselineskip+2em] (Bob) {Bob}
    (Bob.south)
    ++(0,-1)
    node[box] (Willie) {Willie}
  ;

  \coordinate[yshift = #2\baselineskip] (Bob1) at (Bob.west);
  \coordinate[yshift = -#2\baselineskip] (Bob2) at (Bob.west);

  \draw[pfeil] (Pu1) -> (BEC1);
  \draw[pfeil] (Pu2) -> node[above] {$U_2$} (Pxu2);
  \draw[pfeil] (Pxu2.east) -> node[inner sep=0,minimum
          width=4pt,fill,circle,anchor=center] (tap) {} (BEC2);
  \node[above,yshift = 2*#2\baselineskip] (X1) at (tap) {$X_1$};
  \node[above] (X2) at (tap) {$X_2$};
  \draw[pfeil] (tap.center) |- node[above,pos=0.93] {$Z$} (Willie);
  \draw[pfeil] (BEC1.east) -> node[above] {$Y_1$} (Bob1);
  \draw[pfeil] (BEC2.east) -> node[above] {$Y_2$}(Bob2);
  \node[box,fit=(Pu) (Pxu2)] {};

\end{tikzpicture}
}

\ifqtikz
  \exampleChannelESID
\fi

\colorlet{identification}{blue}
\colorlet{transmission}{red}
\tikzset{
	box/.style={draw, minimum height = 1em, minimum width = 1em, inner sep=4pt},
	pfeil/.style={->, >=latex},
  font = {\sffamily\footnotesize}, }

\usepackage{simpleacros}
\ifcompress
  \acsetup{
    first-style = long-short,
    subsequent-style = short,
  }
\else
  \acsetup{
    first-style = long,
    subsequent-style = long,
  }
\fi

\SimpleAcros{
  id = identification,
  bc = broadcast channel,
}

\tryWithFiles\input{local-shorthands,acros}

\begin{document}

\maketitle

\begin{abstract}
                         The problem of identification over a discrete memoryless wiretap channel is  examined under the criterion of semantic effective secrecy. This
  secrecy criterion guarantees both the requirement of semantic secrecy and of stealthy communication.
  Additionally, we introduce the related problem of combining
  approximation-of-output statistics and transmission. We derive a capacity
  theorem for approximation-of-output statistics transmission codes. For a
  general model, we present lower and upper bounds on the capacity, showing that
  these bounds are tight for more capable wiretap channels. We also provide
  illustrative examples for more capable wiretap channels, along with examples of
  wiretap channel classes where a gap exists between the lower and upper bounds.
\end{abstract}

\ifdraft
\ifonecol
  \tryWithFiles \input {edit-notes}
\fi
\fi

\ifdraft
  \tableofcontents
\fi

\section{Introduction}

In modern cyber-physical systems, there is an increasing need for task-oriented
communication paradigms~\cite{Gündüz2023semantic}.
Since Shannon's pioneering work~\cite{Shannon1948MathematicalTheoryCommunication},
traditional communication emphasizes the problem of transmitting
messages, where a decoder is expected to decide which message has been sent over
a noisy channel among exponentially many possible messages.
If the purpose of communication is to achieve a certain task, then this full
reconstruction of a message may be unnecessary. Efficiency gains are
possible by designing communication systems for task-specific requirements
of reliability, robustness, and security~\cite{Shannon1959rdTheory,
  AhlswedeDueck1989Identificationviachannels,AhlswedeZhang1995Newdirectionstheory,
  CabreraEA2021postShannon6G,RezwanCabreraFitzek2022funcomp,Gündüz2023semantic}.

In the task of identification, the receiver, Bob, selects one of many possible
messages and tests whether this particular message was sent by the transmitter,
Alice, or not. It is assumed that Alice does not know Bob's chosen message;
otherwise, she could simply answer with "Yes" or "No" by sending a single bit.
Given the nature of this very specific task, identification is in stark contrast
to the conventional and general task of uniquely decoding messages, i.e.,
estimating which message was sent, as required in transmission. Decoding is
general such that that Bob can estimate any function of Alice's message.
On the other hand, identification only requires Bob to determine whether the message is the one he desires or not, hence the function that Bob is interested in is a simple
indicator function. However, while the code sizes may only grow exponentially in
the block length for the message-transmission task, a \emph{doubly exponential}
growth can be achieved in identification if stochastic encoding is used
\cite{AhlswedeDueck1989Identificationviachannels};
conversely, messages can be compressed by a logarithmic order.
This is possible as stochastic encoding allows the number of messages to be restricted not by the number of possible distinct codewords, but rather by the number of input distributions to the channel that are pairwise distinguishable at the output
\cite{AhlswedeDueck1989Identificationviachannels,
HanVerdu1993Approximationtheoryoutput}.

Identification has applications in various domains that span authentication
tasks such as watermarking \cite{SteinbergMerhav2001Identificationpresenceside,
Steinberg2002Watermarkingidentificationprivate_conference}, sensor communication
\cite{labidi2023joint}, and vehicle-to-X communication
\cite{BocheDeppe2018SecureIdentificationWiretap, BocheArendt2021patentV2Xid,
RosenbergerPeregDeppe2022IdentificationoverCompound_conference}, among others.
From a more theoretical point of view, there are several interesting connections
between identification and common randomness generation
\cite{AhlswedeCsiszar1998Commonrandomnessinformation}, as well as the problem of
approximation of output statistics, also known as channel
resolvability~\cite{HanVerdu1993Approximationtheoryoutput,
  Steinberg1998Newconversestheory,
  Hayashi2006Generalnonasymptoticasymptotic,
  Watanabe2022MinimaxConverseIdentification,
ressurvey}
and soft-covering~\cite[p.~656]{Ahlswede2018CombinatorialMethodsModels},
\cite{Wyner1975commoninformationtwo, Cuff15, BocheDeppe2019SecureIdentificationPassive}.
These connections lead to a remarkable discontinuous behavior of identification
capacities under secrecy constraints. The identification capacity under a
semantic secrecy constraint is the same as the identification
capacity without a secrecy constraint, provided that the secret-transmission
capacity of the channel is positive~\cite{AhlswedeZhang1995Newdirectionstheory, wiretapwafa}.

 Stealthy communication was studied in~\cite{HouKramer14a}, where it was examined
 under a joint requirement of both stealth and secrecy. In stealthy
 communication \cite{cheJaggiEA2014surveyStealth,
 HouKramerBloch2017EffectiveSecrecyReliability,songJaggiEA2020stealthMultipathJammed,BlochGuenlueYenerOggierPoorSankarSchaefer2021OverviewInformationTheoretic},
 the goal is to prevent a potential adversary from discerning whether any
 meaningful communication is occurring or not.
 Covert communication
 \cite{bash2013limits,BashGoeckelTowsleyGuha2015covertLimits,CheBakshiJaggi13,Bloch2016covertResolvability,WangWornellZheng16}
 can be considered as the special case of stealthy communication where the sender
 does not send anything when no meaningful communication takes place
 \cite{LentnerKramer20}. In general, neither stealth nor secrecy implies the
 other, and they can be studied independently. The
 covert setting of the identification problem  for binary-input discrete
 memoryless wiretap channels~\cite{Wyner1975wiretap, BlochBarros2011phys_sec_book}
 has been studied in
 \cite{ZhangTan2021CovertIdentificationBinary}.

Here, we introduce and study the problem of identification under a semantic
effective secrecy constraint. A \textit{semantic} secrecy constraint is one that holds for all possible message probability distributions. An \textit{effective secrecy} constraint is one that combines both secrecy and stealth requirements. A general description
of such a setting is given in \cref{fig:cr-scheme}.
We note that in this model, semantic stealth is equivalent
to semantic effective secrecy and hence does indeed imply semantic secrecy,
in contrast to the setting of strong secrecy and stealth
\cite{Hou14thesis,HouKramer14a,HouKramerBloch2017EffectiveSecrecyReliability},
which is defined as an average case over the messages.
We obtain both lower and upper bounds on the capacity which
are are tight for some classes of discrete memoryless wiretap channels. In
particular, we prove the capacity for the class of more capable discrete memoryless
wiretap channels. For other channels,
the derived achievable rate bound is further enhanced by
observing that the encoder can always prepend an auxiliary channel to the actual
channel, which may increase the achievable rates for some classes of channels.
As a tool, we introduce the problem of joint approximation of output statistics
at one party and transmission of information to another party.
For this task, we determine the capacity, as an auxiliary result to prove
the identification capacity with effective secrecy: in our identification code
construction, stealth is ensured by approximating the innocent output
statistics expected by the attacker, while reliable identification is ensured by an
identification code for the noiseless channel built upon the transmission
part of the underlying approximation-of-output-statistics code.

To illustrate our results, we provide multiple examples.
One compares the capacity behavior of approximation-of-output-statistics transmission capacity and a simple
transmission capacity, as well as comparing both with the optimal rate for pure approximation of output
statistics without transmission.
To provide intuition for our bounds on the identification capacity with
effective secrecy, we discuss several examples of channels where this capacity is known.
These include a wiretap channel with reversely degraded subchannels, where the capacity is
known only for a certain range of channel parameter values,
whereas for different parameters, the upper and lower bounds do not coincide.

Finally, we discuss future steps needed to obtain a generally tight converse
bound. The identification capacity of the discrete memoryless wiretap channel
exhibits a similar dichotomy for semantic effective secrecy as for only semantic
secrecy \cite{AhlswedeZhang1995Newdirectionstheory}, but with a more stringent
positivity condition and constraint. This is because for secrecy, only a small
part of the identification codeword has to be secret
\cite{AhlswedeZhang1995Newdirectionstheory}, while for effective secrecy, the
whole codeword must be stealthy
\cite{IbrahimFerraraDeppe2021IdentificationEffectiveSecrecy_conference}.

This paper is organized as follows: \cref{sec:prelim} introduces our notation
and provides the preliminaries and necessary background.
In \cref{sec.code.denfns}, we define
several classes of codes that are relevant to the problems tackled in this
paper, including both approximation-of-output-statistics-transmission-codes
(\cref{AOSTcodeDef})
and effectively secret identification codes (\cref{ESID.code.defn}).
In \cref{sec:results}, we present our
main results, and in \cref{sec:example} we discuss various examples.
\Cref{sec:esid.dm.proof} contains the proofs of the theorems in \cref{sec:results}.
Finally, \cref{sec:conclusion} summarizes the results and discusses further steps.

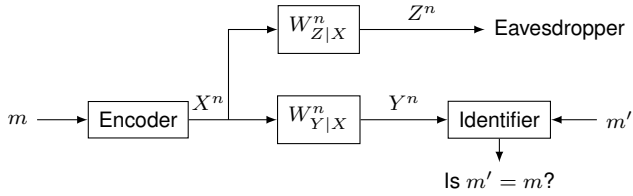
\begin{figure}
  \centering
  \begin{tikzpicture}[scale=1.6]
    \path
        node[box] (channel) {$W_{Y|X}^n$}
        ++(0,.75)
        node[box] (wiretap) {$W_{Z|X}^n$}
        ++(2,0)
        node      (eve)     {Eavesdropper}
        (channel)
        ++(-.75,0)
        coordinate (tap)
        (channel)
        ++(-1.5,0)
        node[box] (encoder) {Encoder}
        ++(-1,0)
        node      (message) {$m$}
        (channel)
        ++(1.5,0)
        node[box] (decoder) {Identifier}
        +(0,-.5)
        node      (test)    {Is $m' = m$?}
        ++(1,0)
        node      (rec_message) {$m'$}
    ;

    \draw[pfeil] (message) -> (encoder);
    \draw[pfeil] (encoder) -- node[above] {$X^n$} (tap) -> (channel);
    \draw[pfeil] (tap) |- (wiretap);
    \draw[pfeil] (channel) -> node[above] {$Y^n$} (decoder);     \draw[pfeil] (decoder) -> (test);
    \draw[pfeil] (rec_message) -> (decoder);
    \draw[pfeil] (wiretap) -> node[above] {$Z^n$} (eve);
    \end{tikzpicture}
	\caption{\label{fig:cr-scheme}
	  A general identification scheme is considered in the presence of an eavesdropper, where $(W_{Y|X}, W_{Z|X})$ is a discrete memoryless wiretap channel (DM-WTC). In contrast to conventional message transmission, the receiver does not decode the message $m$ from the channel output $Y^n$. Instead, the receiver selects a message $m'$ and performs a statistical hypothesis test to decide whether $m'$ equals $m$ or not. In the effective secrecy setting, the eavesdropper aims to determine whether unexpected communication is occurring compared to some expected default behavior, and to identify whether its own message $m'$ equals $m$ or not.
	}
\end{figure}

\section{Background and Preliminaries}
\label{sec:prelim}

\subsection{Notation}

The following notation is also summarized in \cref{tbl:notation}.
Calligraphic letters $\mathcal{X}, \mathcal{Y}, \mathcal{Z}, \dots$ denote finite sets.
Capital letters $X, Y, Z, \dots$ denote random variables (RVs),
while small letters $x, y, z, \dots$ denote realizations of random variables.
The set of all possible probability distributions on a finite alphabet $\mathcal{X}$
  is denoted by $\mathcal{P}(\mathcal{X})$.
A probability mass function (PMF) of a random variable $X$ is denoted by $P_X$,
where $P_X \in \mathcal{P}(\mathcal{X})$.
A joint PMF of a pair of RVs $X$ and $Y$ is denoted by $P_{X,Y}$,
and a conditional PMF by $P_{Y|X}(y|x) \coloneqq P_{X,Y}(x,y) / P_X(x)$.
We say that $X \sim P_X$ and $(X,Y) \sim P_{X,Y}$,
  if $\Pr(X=x) = P_X(x),\, \forall x\in\mathcal{X}$,
  and $\Pr(X=x, Y=y) = P_{X,Y}(x, y),\, \forall (x,y)\in\mathcal{X}\times \mathcal{Y}$,
  respectively.
The PMF $P_Y \coloneqq P_X \circ P_{Y|X}$ is a probability distribution on $\mathcal{Y}$ where
\begin{align}
(P_X \circ P_{Y|X})(y) \coloneqq \sum_{x \in \mathcal{X}} P_X(x) P_{Y|X}(y|x).
\end{align}
Extending this terminology, the PMF $P_Z \coloneqq P_X \circ P_{Y|X} \circ P_{Z|X,Y}$
is a probability distribution on $\mathcal{Z}$ where
\begin{align}
  (P_X \circ P_{Y|X} \circ P_{Z|X,Y})(z) \coloneqq \sum_{(x, y) \in (\mathcal{X}, \mathcal{Y} ) } P_X(x) P_{Y|X}(y|x) P_{Z|X,Y}(z|x,y).
\end{align}
$P_X \times W_{Y|X}$ refers to a joint PMF on $(X,Y)$ such that $P_X \times W_{Y|X}(x,y) = P_X(x) W_{Y|X}(y|x)$.

$X_{i}^{j}$ is used to refer to a sequence of random variables, i.e.,
$X_{i}^{j} = (X_i, \dots, X_j)$ for all $i \le j$. $X^n$ denotes a sequence of
RVs such that $X^n = (X_1, \dots, X_n)$. The sequence of RVs
$X^{n_1+n_2} = (X^{n_1}, X_{n_1+1}^{n_1+n_2})$ is created by appending those two
sequences of RVs.

A sequence of RVs, $X^n$, is said to be \textit{i.i.d.} (independent and identically distributed) if their joint PMF can be written as
 $P_{X^n}(x^n) \coloneqq P_X^n(x^n) = \prod_{i=1}^{n} P_X(x_i)$. If RVs $X$ and $Z$ are conditionally independent given $Y$, i.e., $P_{X,Z|Y} = P_{X|Y} P_{Z|Y}$, then $X, Y, Z$ are said to form a \textit{Markov chain} $X - Y - Z$. $P^{\mathbf{unif}}_{\mathcal{A}}$ refers to a uniform probability distribution over set $\mathcal{A}$.

Given sequences $X^n\in\mathcal{X}^n$ $P_X \in \mathcal{P}(\mathcal{X})$, define the $\epsilon$-typical set with respect to a $P_X \in \mathcal{\mathcal{X}}$
\begin{align}
\mathcal{T}_{\epsilon}^{n}(P_X) = \left\{ x^n \in \mathcal{X}^n : \left| \frac{1}{n} N(a|x^n) - P_X(a) \right| \leq \epsilon \, P_X(a), \, \forall a \in \mathcal{X} \right\},
\end{align}
where $N(a|x^n)$ is the number of occurrences of symbol $a$ within sequence $x^n$. The set of jointly $\epsilon$-typical sequences $\mathcal{T}_{\epsilon}^{n}(P_{X,Y})$ can be defined likewise. The indicator function $\ind{\cdot}$ evaluates to $1$ if its
argument is true, and to $0$ if it is false. All logarithms are in base $2$
unless explicitly stated otherwise.

\subsection{Information Measures and Divergences}
\label{sec.divergences}
In this subsection we define information measures and statistical divergences and the relationships and rules applying to them. These are widely used throughout the paper.
Shannon entropy is defined as $H(X)=-\expect_{P_X}\brack{\log{P_X(X)}}$.
The support of a PMF $P_X \in \cP(\cX)$ is defined as the set $\supp P_X = \set{x\in\cX|\,P_X(x) > 0}$.
\begin{definition}[KL-divergence]
For two PMFs $P_X,Q_X \in \cP(\cX)$, where  $\supp Q_X = \cX$,
the \defname[KL-divergence]{Kullback-Leibler divergence}
is defined by
\begin{gather}
 D(P_X \| Q_X) =
     \expect_{P_X}\brack{\log\frac{P_X(X)}{Q_X(X)}}= \sum\limits_{x\in\mathcal{X}} P_X(x) \log \frac{P_X(x)}{Q_X(x)} .
    \end{gather}
Note the restriction to $x$ for which the denominator $Q_X(x) > 0$.
\end{definition}
\begin{definition}[Conditional KL-divergence]
For two conditional PMFs $W_{Y|X}, V_{Y|X}$,
the conditional KL-divergence is defined as
\begin{align}
D(W_{Y|X} \| V_{Y|X} | P_X)
  &\coloneqq \expect_{P_{X}} \brack{ D(W_{Y|X} \| V_{Y|X})} \nonumber \\
  &= \sum\limits_{x\in\mathcal{X}}\sum\limits_{y\in\mathcal{Y}}P_{X}(x)W_{Y|X}(y|x)\log\frac{W_{Y|X}(y|x)}{V_{Y|X}(y|x)}
  \,,
\end{align}
where $D(W_{Y|X} \| V_{Y|X})=\sum_{y\in\mathcal{Y}} W_{Y|X}(y|X)\log\frac{W_{Y|X}(y|X)}{V_{Y|X}(y|X)}$.
\end{definition}
\begin{definition}[Mutual Information]
The mutual information is defined as follows:
\begin{align}
I(X;Y)&:=I(P_X;W_{Y|X}):=D(P_{X,Y}\|P_X P_Y) \nonumber \\
&=\expect_{P_{X,Y}}\brack{\log\frac{P_{X,Y}(X,Y)}{P_X(X)P_Y(Y)}} \nonumber \\
&=D(W_{Y|X}\|P_{Y}|P_X),
\end{align}
where $X\sim P_X$ and $(X,Y)\sim P_X\times W_{Y|X}$.
\end{definition}

A useful chain rule for the KL-divergence enables us to decompose the KL-divergence between two joint distributions $P_{X,Y}$ and $Q_{X,Y}$ where $P_{X,Y}= P_X \times W_{Y|X}$ and  $Q_{X,Y}= Q_X \times V_{Y|X}$
\begin{align}
  D(P_{X,Y}\|Q_{X,Y})
  &= \expect_{P_{X,Y}}\brack{\log\frac{P_{X,Y}(X,Y)}{Q_{X,Y}(X,Y)}} \nonumber  \\
  &= \sum\limits_{x\in\mathcal{X}}P_X(x)\log\frac{P_X(x)}{Q_X(x)}+\sum\limits_{x\in\mathcal{X}}\sum\limits_{y\in\mathcal{Y}}P_X(x)P_{Y|X}\log\frac{P_{Y|X}(y|x)}{Q_{Y|X}(y|x)} \nonumber \\
  &=D(P_X\|Q_X)+D(W_{Y|X}\|V_{Y|X}|P_X)
  \,\fxcolor{teal}.
  \label{KLchainrule}
\end{align}

Note that for $(X,Y) \sim P_X \times W_{Y|X}$ and any $Q_Y \in \cP(\cY)$
such that the following expression is defined, it holds that
\begin{align}
  &D(W_{Y|X} \| Q_Y | \, P_X) \nonumber \\
  &= \expect_{P_{X,Y}} \brack{ \log \frac {W_{Y|X}(Y|X)} {Q_Y(Y)} } \nonumber \\
  &= \expect_{P_{X,Y}} \brack{ \log \frac {W_{Y|X}(Y|X)} {(P_X\circ W_{Y|X})(Y)} + \log \frac {(P_{X}\circ W_{Y|X})(Y)} {Q_Y(Y)} } \nonumber \\
  &= D(W_{Y|X} \| P_X\circ W_{Y|X} | P_X) + D(P_{X}\circ W_{Y|X} \| \, Q_Y) \nonumber \\
  &= I(X;Y) + D(P_{X}\circ W_{Y|X} \| Q_Y)
  \,.
\end{align}

\begin{definition}[Binary entropy function and binary KL-divergence]
For a binary RV $X$,
where $\mathcal{X}=\{0,1\}$ and $P_X(1)=p$,
the binary entropy function is defined as
$H_b(p)=-p\log p-(1-p)\log(1-p)$.
The KL-divergence between two Bernoulli RVs with alphabet $\mathcal{X}=\{0,1\}$, where $P_X(X=1)=p$ and $Q_X(X=1)=q$, is $D_{b}(p \, \| \, q)=p\log\frac{p}{q}+(1-p)\log\frac{1-p}{1-q}$.
\end{definition}
The following lemma is well-known and used for our analysis.
\begin{lemma}[{A Chernoff bound \cite[Theorem 2.1.]{Mulzer18}{}}]
\label{chern}
Let $X_1,\dots, X_n$ be i.i.d  Bernoulli RVs with $\Pr(X_i=1)= p_i$, $0 \leq p_i \leq 1$. Set $X = \frac{1}{n}\sum_{i=1}^{n} X_i$, $p = \frac{1}{n}\sum_{i=1}^{n} p_i$ and $t\in (0, 1-p]$, then
\begin{align}
\text{Pr}(X \geq p+t) \leq 2^{- n \, D_b(p+t \, \| \, p)} .
\end{align}
\end{lemma}

\subsection{Channels}
\begin{definition}[Channel]
A \defname{channel} with domain $\cX$ and codomain $\cY$
is a conditional PMF $W_{Y|X}$ associated with an input alphabet $\mathcal{X}$ and an output alphabet $\mathcal{Y}$.
\end{definition}
\begin{definition}
A \defname[DMC]{discrete memoryless channel} $W_{Y|X}$ is a sequence
$(W_{Y|X}^n)_{n \in \bbN}$, where an input sequence $x^n$
of block length $n \in \bbN$
is mapped to an  output sequence $y^n$
with probability
$
  W_{Y|X}^n(y^n|x^n)\coloneqq  \prod\limits_{i=1}^n W_{Y|X}(y_i | x_i).
$
\end{definition}
The output distribution of sequence $Y^n$ resulting from sending an input symbol $X$ that is distributed according to $P_X$ over channel $W_{Y|X}$ where $ P_{X}\circ W_{Y|X}(y)\coloneqq P_Y(y)=\sum\limits_{x\in\mathcal{X}}P_X(x)W_{Y|X}(y|x)$.

One particularly important class of binary-input channels is the binary symmetric channel $\BSC(\epsilon)$,
with a binary input and output alphabet $\mathcal{X}=\mathcal{Y}=\{0,1\}$,
described by the channel probability transition matrix $
  \BSC(\epsilon)=
  \mleft(\begin{smallmatrix}
    1 - \epsilon \,  &\, \epsilon \\
     \epsilon \, & \, 1 - \epsilon
  \end{smallmatrix}\mright)
  $,
  where the entries of each row of the channel matrix are the probabilities of each possible output symbol given a fixed input symbol per row. Hence, each row in the channel matrix should sum to one.

Another important example of binary-input channels is the binary erasure channel
$\BEC(\epsilon)$ with a binary input alphabet $\mathcal{X}=\{0,1\}$ and a ternary
output alphabet $\mathcal{Y}=\{0,e,1\}$ where  $e$ represents a lost symbol through
the channel that can no longer identified as either being a $0$ or a $1$,
and the channel probability transition matrix
$
  \BEC(\epsilon)=
  \tup{\begin{smallmatrix}
    1 - \epsilon\,  &\, \epsilon \, & \, 0 \\
    0 \,  & \, \epsilon\, &\, 1 - \epsilon
  \end{smallmatrix}}
$.
\begin{definition}[Wiretap Channel]
A \defname{wiretap channel} (WTC) is a channel $W_{Y,Z|X} : \cX \to \cP(\cY \times \cZ)$,
where we assume that for an input $x$ at a sender ($\textit{Alice}$),
a legitimate receiver (\textit{Bob}) has access to the channel output $Y \sim W_{Y|X=x}$
and a passive adversary (\textit{Willie}\footnote{
We call the adversary Willie instead of Eve, following \cite{Hou14thesis,HouKramer14a,HouKramerBloch2017EffectiveSecrecyReliability},
because the focus of this paper is on a stealthy setting,
where the adversary is a warden checking if everything is in order.})
has access to an output $Z \sim W_{Z|X=x}$,
where the output distributions of the marginal channels $W_{Y|X=x}$ and $W_{Z|X=x}$
given input $x$ are marginalizations of the joint output distribution $W_{Y,Z|X=x}$.
\end{definition}
A \defname[DM-WTC]{discrete memoryless wiretap channel} $W_{Y,Z|X}$ is
a sequence $(W_{Y,Z|X}^n)_{n \in \bbN}$.

We usually refer to DM-WTC by an ordered pair of marginal channels
$(W_{Y|X}, W_{Z|X})$, where the first
is a point-to-point channel from Alice to Bob.
and the second is a point-to-point channel from Alice to Willie.

One way to characterize the statistical advantage of one user over another is to
use the notion of \textit{degradedness}.

\begin{definition}[Stochastically degraded DM-WTC]
A DMC-WTC $W_{Y,Z|X}$ with marginals $(W_{Y|X}, W_{Z|X})$ is \defname{stochastically degraded}
if there exists a conditional distribution $\tilde{W}_{Z|Y}$ such that
\begin{align}
W_{Z|X}(z|x)=\sum_{y\in\mathcal{Y}} W_{Y|X}(y|x) \tilde{W}_{Z|Y}(z|y),\,\,\forall \, x\in\mathcal{X},\,z\in\mathcal{Z}.
\end{align}
Given a pair $(X,Y)\sim P_X\times W_{Y|X}$, and another pair
 $(X,Z)\sim P_X\times W_{Z|X}$.
 \end{definition}
\begin{definition}[More capable DM-WTC]
 A DM-WTC  $(W_{Y|X},W_{Z|X})$
is said to be \textit{more capable} if
\begin{align}
I(X;Y)\geq I(X;Z) ,
\end{align}
for all $P_X\in\mathcal{P}(\mathcal{X})$.
\end{definition}
Finally, A DM-WTC $(W_{Y|X},W_{Z|X})$ is said to be \textit{less noisy} if
\begin{align}
I(U;Y)\geq I(U;Z) ,
\end{align}
for all $P_{U,X}\in\mathcal{P}(\mathcal{U}\times\mathcal{X})$ such that $U - X - (Y,Z)$. Note that every less noisy DM-WTCs is a more capable DMC-WTC as well, but the converse is not true. A DMC is said to be  \textit{degenerate} if the rows of the channel probability transition matrix are linearly dependent.

\subsection{Secrecy Notions}

Information-theoretic study of secrecy in communication problems was pioneered by Shannon in \cite{Shannon49a}. For a communication system where a sender \textit{Alice} sends a message $m \in \mathcal{M}$ through codeword $X^n$ to a receiver \textit{Bob} while the communication is being observed by an eavesdropper \textit{Willie} who receives $Z^n$.  The received sequence $Z^n$ at Willie and the message $M$ are required to be statistically independent. In such a system, the output distribution $P_Z^n$ takes the following form
\begin{align}
  P_{Z^n}(z^n)=\sum\limits_{m,x^n}P_{M}(m)P_{X^n|M}(x^n|\,m)W_{Z^n|\,X^n}(z^n|x^n).
\end{align}

\begin{definition}[Perfect secrecy]
In Shannon's model, a very strong notion of secrecy was studied which is
\textit{perfect secrecy}. A coding scheme is said to achieve perfect secrecy if
the \textit{information leakage} $L^{(n)}\coloneqq I(M; Z^n)$ at Willie satisfies
\begin{align}
L^{(n)} =0.
\end{align}
\end{definition}
The secrecy capacity of the degraded discrete memoryless wiretap channel
(DM-WTC) was derived by Wyner, who introduced the wiretap model in
\cite{Wyner1975wiretap} and studied it under a relaxed secrecy criterion, namely the \textit{weak secrecy} criterion. That model assumed a degraded $W_{Y,Z|X} \coloneqq
W_{Y|X} W_{Z|Y}$ to describe the statistical relation between input $X$ by
Alice, outputs $Y$, $Z$ at Bob and Willie, respectively.
\begin{definition}[Weak secrecy]
In a relaxation of the often very stringent perfect secrecy condition, weak
secrecy was used as a measure of secrecy. To achieve weak secrecy, the
normalized information leakage at Wille should asymptotically converge to zero
\begin{align}
\lim\limits_{n\to\infty}\frac{1}{n}L^{(n)}=0 .
\end{align}
\end{definition}

\begin{definition}[Strong secrecy]
Another asymptotic notion of secrecy that is more stringent than weak secrecy
but still less stringent than perfect secrecy is \textit{strong secrecy}. To
achieve strong secrecy, the unnormalized information leakage at Willie should
asymptotically converge to zero
\begin{align}
\lim\limits_{n\to\infty}L^{(n)}=0 .
\end{align}
\end{definition}

\begin{definition}[Semantic Secrecy]
\textit{Semantic secrecy} is yet another asymptotic secrecy criterion that is more
stringent than either weak secrecy or strong secrecy but less stringent than
perfect secrecy. To achieve semantic secrecy, the worst-case unnormalized
information leakage at Willie should asymptotically converge to zero
\begin{align}
  \lim\limits_{n\to\infty}\max\limits_{P_M \in \mathcal{P}(\mathcal{M})} L^{(n)} = 0
  \,,
\end{align}
where the worst case is considered with respect to all possible message distributions.
\end{definition}

\subsection{Stealth and Covertness}
\label{seq:securityMeasures.stealth}

The aforementioned secrecy measures quantify the statistical dependence
between a message and a received sequence at a potential adversary Willie.
On the other hand, a \emph{stealth} constraint should guarantee that the
adversarial warden Willie cannot detect if Alice sent a meaningful signal or not.
Hence, a good stealth measure should quantify the statistical
distinguishability between these two modes of communication:
The first mode corresponds to the probability distribution $P_Z^n$ at Willie's channel
output when useful information is being sent,
whereas the second mode corresponds to the \textit{reference distribution} $Q_{Z^n}$
that is expected by Willie at the channel output when no useful information is
being sent. This $Q_{Z^n}$ need not be a possible output distribution of the
channel; it can be freely chosen by Willie, but we assume it is known by Alice and Bob.

As a measure $\Delta^{(n)}$ for the statistical distinguishability of $P_{Z^n}$ and $Q_{Z^n}$, we use
\begin{align}
  \Delta^{(n)} \coloneqq D(P_{Z^n}\|\, Q_{Z^n})
  \,.
\end{align}
By combining the leakage and distinguishability, we define the \emph{effective-secrecy}
leakage $L_{\text{eff.}}^{(n)} \coloneqq L^{(n)}+\Delta^{(n)}$,
which can be expressed by a single KL-divergence as follows:
\begin{align}
L_{\text{eff.}}^{(n)}
&= L^{(n)}+\Delta^{(n)}\label{eq:effective secrecy} \nonumber \\
&= I(M; Z^n)+ D(P_{Z^n}\| \, Q_{Z^n}) \nonumber\\
&= D(P_{M,Z^n}\| \, P_{M}Q_{Z^n})
.
\end{align}
As we are interested in reliability for the worst-case message,
we consider \emph{semantic effective secrecy},
where the effective-secrecy leakage should asymptotically converge to zero
 \begin{align}
\lim_{n\to\infty} \max_{P_M \in \mathcal{P}(\mathcal{M})} L^{(n)}_{\text{eff.}}&=\lim_{n\to\infty} \max_{P_M \in \mathcal{P}(\mathcal{M})}  D(P_{M,Z^n}\| \, P_{M}Q_{Z^n}) \nonumber \\
&=\lim_{n\to\infty} \max_{P_M \in \mathcal{P}(\mathcal{M})}\expect_{P_M}\brack{D(P_{Z^n|M}\| \, Q_{Z^n})}\nonumber \\
&=\lim_{n\to\infty} \max_{m\in \mathcal{M} }D(P_{Z^n|M=m}\| \, Q_{Z^n}) \nonumber \\
&= 0.
\label{eq:effSecrecyLeakageBound}
\end{align}
In this paper, we focus on the case where this reference distribution is i.i.d., i.e. $Q_{Z^n} = Q_Z^n$.

\begin{remark}
  \label{remark:semanticStealthEquivSemanticEffectiveSecrecy}
  A pure semantic stealth constraint would be defined by
  \begin{equation}
    \lim_{n \to \infty} \max_{m \in \cM} \Delta^{(n)}_m = 0
    \,.
    \label{eq:semanticStealthBound}
  \end{equation}
  However, from \cref{eq:effSecrecyLeakageBound} one can see that
  \begin{equation}
    \max_{m \in \cM} \Delta^{(n)}_m
    = \max_{m \in \cM} D(P_{Z^n|M=m} \| Q_{Z^n})
    = \max_{P_M \in \cP(\cM)} L_{\text{eff.}}^{(n)}
    ,
  \end{equation}
  and hence semantic stealth is equivalent to semantic effective secrecy
  and implies semantic secrecy.
  This is why the definitions of
  effectively secret transmission (\cref{ESTcodes.defn})
  and identification codes (\cref{ESID.code.defn})
  contain actually the semantic stealth
  constraint~\cref{eq:semanticStealthBound}.
  This is in contrast to strong stealth and secrecy, where neither of them implies the
  other~\cite{Hou14thesis,HouKramer14a,HouKramerBloch2017EffectiveSecrecyReliability}.
\end{remark}

\begin{center}
\begin{DIFnomarkup} \begin{longtblr}[
  label = {tbl:notation},
  caption={General notation and acronyms},
]
{
    colspec={l X l}, hline{1,Z} = {solid},
    column{1}={$},
}
  X,Y,\dots
      && random variables (RVs)
      \\
  \cX,\cY, \dots
      && finite sets (alphabets)
      \\
  x,y,\dots
      && realizations of RVs
      \\
                    x^n = (x_1,x_2,\dots,x_n) \in \cX^n
      && sequence of length $n$
      \\
  X_k^n = (X_k,X_{k+1},\dots,X_n)       && sequence of RVs of length $(n-k+1)$
      \\
  P_X && probability mass function (PMF) of $X$
      \\
  \expect_{P_X}(f(X))       && expectation of an RV $f(X)$ with distribution $P_X$
      \\
                    \cP(\cX)
      && set of all PMFs with finite support over a set $\cX$
      \\
  P_{X}^n(X^n) = \prod_{i=1}^n P_X(X_i)
      && $n$-fold product distribution
      \\
        \text{SID} && secret identification
  \\
  \text{ESID} && effectively secret identification
  \\
  \text{ST} && secret transmission
  \\
  \text{EST} && effectively secret transmission
  \\
  \text{AOS} && approximation of output statistics
  \\
  \text{AOST} && approximation of output statistics and transmission
  \\
  \capT(W_{Y|X}) &&  transmission capacity of a point-to-point channel
  \\
  \capST(W_{Y|X},W_{Z|X}) && ST capacity of a DM-WTC
  \\
  \capID(W_{Y|X}) &&  ID capacity of the point-to-point channel
  \\
  \capSID(W_{Y|X},W_{Z|X}) &&  SID capacity
  \\
  \capESID(W_{Y|X},W_{Z|X},Q_Z) && ESID capacity with an i.i.d. reference distribution $Q_Z^n$
  \\
  \capEST(W_{Y|X},W_{Z|X},Q_Z) && EST capacity
  \\
  \capAOST(W_{Y|X},W_{Z|X},Q_Z) && AOST capacity
  \\
  \mathbf{R^{\star}_{AOS}}(W_{Y|X},Q_Y) && optimal AOS rate
                                \\
            
\end{longtblr}
\end{DIFnomarkup}
\end{center}

\section{Code Classes}
\label{sec.code.denfns}
In this section, we formally define several classes of codes that are relevant to our work here. These serve as building blocks to our investigations. We also aim to contextualize the novel problems of considered in this paper within the pre-existing literature.

\subsection{Transmission Codes, With and Without Secrecy Constraints}
We begin by a definition of transmission codes that will be used and built upon on in the rest of the paper.

\begin{definition}[Transmission Code]
  A $(|\mathcal{M}_{\text{T}}|, \, n \, |\,\lambda^{(n)})$-transmission code for a DMC $W_{Y|X}$
  consists of the following:
  \begin{enumerate}
    \item A message set $\mathcal{M}_{\text{T}}$,
    \item A stochastic encoder represented by the function $f_n \coloneqq P_{X^n|M}:\mathcal{M}_{\text{T}}\rightarrow \mathcal{P}(\mathcal{X}^n)$,

    \item a decoder function $\phi_n:\mathcal{Y}^n\rightarrow \mathcal{M}_{\text{T}}$.
  \end{enumerate}
  {$E_m$}
  is used a shorthand for $P_{X^n|M=m}$. It represents the conditional probability distribution of the channel input $X^n$ given a particular message $m$.
  To transmit a message $m\in \mathcal{M}_{T}$, the encoder picks $X^n(m) \sim E_m$ and sends it over the channel $W_{Y|X}$.
 $X^n(m)$ is the input
  sequence over $n$ uses of the DMC $W_{Y|X}$, while $y^n$ is the output
  sequence following the distribution
  $W_{Y|X}^n(Y^n|x^n(m))$. A decision (decoding) set for message $m$ is defined such that
  $\mathcal{D}_m=\{y^n: \phi(y^n)=m\}$.
  A useful notation that we use is to refer to transmission code is by the set of pairs
  \begin{equation}
  \big\{(E_m,\mathcal{D}_{m})\big\}_{m\in\mathcal{M}_{\text{T}}}.
  \end{equation}

  At the decoder, the following reliability condition should be satisfied
  \begin{align}
  \max_{m\in\mathcal{M}_{\text{T}}} E_{m} \circ W_{Y|X}^n(\mathcal{D}^{c}_{m})\leq \lambda^{(n)}.
  \end{align}

  A rate $R_{\text{T}}\coloneqq \frac{1}{n}\log{|\mathcal{M}_{\text{T}}}|$ is said to be
  \defname{achievable} if there exists a sequence
  of $(|\mathcal{M}_{\text{T}}|, \, n \, |\,\lambda^{(n)})$-transmission codes such that
  $\lim_{n \to \infty} \lambda^{(n)} = 0$.
  The transmission capacity $\capT(W_{Y|X})$ is the supremum of all achievable rates. It was found that $\capT$ is given by \cite{shannon}
  {
  \begin{equation}
    \capT(W_{Y|X})=\max_{P_X \in \cP(\cX)} I(X;Y)
    \,.
  \end{equation}
  }
\end{definition}

\begin{definition}[{Secret Transmission (ST) Code~\cite{Wyner1975wiretap,CsiszarKorner78}}]
  A $(|\mathcal{M}_{\mathsf{ST}}|, \, n \, | \, \lambda^{(n)},\delta^{(n)})$-ST code for a DM-WTC $(W_{Y|X},W_{Z|X})$
  is a $(|\mathcal{M}_{\text{T}}|, \, n \, | \, \lambda^{(n)})$-transmission code which satisfies
  the semantic secrecy criterion
  \begin{align}
    \label{ST.secrecy.requirements}
    \max_{P_M \in \mathcal{P}(\mathcal{M}_{\mathsf{ST}})}I(M;Z^n)
    \le \delta^{(n)}
    .
  \end{align}
  Similarly to transmission codes, a rate $R_{\text{T}}\coloneqq \frac{1}{n}\log{|\mathcal{M}_{\mathsf{ST}}}|$ is achievable
  if there exists a sequence of $(2^{n R_{\mathsf{ST}}}, n \, |\, \lambda^{(n)}, \delta^{(n)})$-ST codes
  such that
  \begin{equation}
  \lim_{n \to \infty} \lambda^{(n)} = \lim_{n \to \infty} \delta^{(n)} = 0.
  \end{equation}
  The ST capacity $\capST(W_{Y|X}, W_{Z|X})$ is the supremum of all achievable rates. It was found that $C_{\mathsf{ST}}$
   is ~\cite{Wyner1975wiretap,CsiszarKorner78}
  \begin{equation}
    \capST(W_{Y|X}, W_{Z|X})
    = \max_{P_{U,X} \in \Pst} I(U;Y) - I(U;Z)
    \,,
  \end{equation}
  where $\Pst = \set{
    P_{U,X}\in \cP(\cU \times \cX)
    : \card \cU \le \card \cX
    ,\, U - X - (Y, Z)
  }$.
\end{definition}

\begin{remark}
  To achieve positive secrect transmssion rates, it can be shown that stochastic encoding is essential for secret transmission codes \cite{Wyner1975wiretap, CsiszarKorner78}.
  But stochastic encoding does not increase the transmission capacity of a DMC~\cite[p.~73]{ElGamalKim2011NetworkInformationTheory}. Thus, stochastic encoding is not necessary for transmission codes without secrecy constraints. A deteminstic mapping $f_n: \mathcal{M_{T}}\rightarrow \mathcal{X}^n$ can be used for encoding while not reducing $\capT$.
\end{remark}

\subsection{Identification Codes, With and Without Secrecy Constraints}

In this subsection we define identification codes with and without secrecy constraints.

\begin{definition}[{Identification (ID) codes (see~\cite{AhlswedeDueck1989Identificationviachannels})}]

A $(|\mathcal{M}_{\text{ID}}|,n\,|\,\lambda_1^{(n)},\lambda_2^{(n)})$-ID code for a DMC $W_{Y|X}$ consists of
the following:
\begin{enumerate}
  \item A message set $\mathcal{M}_{\text{ID}}$,
 \item A stochastic encoder represented by the function $f_n:\mathcal{M}_{\text{ID}}\rightarrow \mathcal{P}(\mathcal{X}^n)$. $f_n$ is governed by a conditional
  probability distribution $P_{X^n|M}$,
  \item a deterministic identifier function
    $g_n : \cY^n \times \cM_{\text{ID}} \to \set{ 0,1 }$. That function defines for message $m$ a decision
    set $\mathcal{I}_{m}$ such that $\mathcal{I}_m = \set{ y^n: g(y^n,m)=1 }$.
\end{enumerate}
As in transmission, $E_m$ is used as a shorthand for $P_{X^n|M=m}$ and represents the conditional probability distribution of the channel input $X^n$ given a particular message $m$.
We usually refer to ID code by the encoder-identifier pairs
  \begin{equation}
  \Big\{\big(E_m,\mathcal{I}_{m}\big)\Big\}_{m_\in\mathcal{M}_{\text{ID}}}.
  \end{equation}
For reliability, the pairs $\set{ E_m, \cI_m }_{m \in \cM_{\text{ID}}}$ must satisfy
\begin{align}
\max\limits_{m\in \mathcal{M}_{\text{ID}}}E_{m}\circ W_{Y|X}^n(\mathcal{I}^c_m)\leq \lambda_1^{(n)},\\
\max\limits_{m\in\mathcal{M}_{\text{ID}},m':m\neq m'}E_{m}\circ W_{Y|X}^n(\mathcal{I}_{m'})\leq \lambda_2^{(n)}.
\end{align}
A rate $R_{\text{ID}}\coloneqq\frac{1}{n}\log\log{|\mathcal{M}_{\text{ID}}|}$ is said to be \textit{achievable} if
there exists a sequence of $(|\mathcal{M}_{\text{ID}}|,n\,
$ $|\,\lambda_1^{(n)},\lambda_2^{(n)})$-ID codes such that
\begin{equation}
\lim_{n\rightarrow\infty}\lambda_1^{(n)}=  \lim_{n\rightarrow\infty}\lambda_2^{(n)}=0 .
\end{equation}
The ID capacity $\capID(W_{Y|X})$ is defined as the supremum of all achievable rates. It was found that $\capID$
is given by~\cite{AhlswedeDueck1989Identificationviachannels,Ahlswede2008Generaltheoryinformation}
\begin{equation}
  \capID(W_{Y|X}) = \capT(W_{Y|X}) =  \max_{P_X \in \cP(\cX) } I(X;Y)
  \,.
\end{equation}
\end{definition}

\begin{definition}[{Secret identification (SID) code~\cite{AhlswedeZhang1995Newdirectionstheory}}]
  A $(|\mathcal{M}_{\text{SID}}|,n\,|\,\lambda_1^{(n)},\lambda_2^{(n)},\delta^{(n)})$-SID code
  for a WTC $(W_{Y|X},W_{Z|X})$
  is a $(|\cM_{\text{SID}}| , n \,|\, \lambda_1^{(n)}, \lambda_2^{(n)})$-ID code for $W_{Y|X}$
  that satisfies the additional semantic secrecy requirement
  \begin{align}
    \max_{P_M\in \mathcal{P}(\mathcal{M}_{\text{SID}})} I(M;Z^n) \leq \delta^{(n)}
    \,,
  \end{align}
  Similar to identification codes,
  a rate $R_{\text{SID}}\coloneqq\frac{1}{n}\log\log|\mathcal{M}_{\text{SID}}|$
  is achievable if there is a sequence of
  $(|\cM_{\text{SID}}|, n \,|\, \lambda_1^{(n)}, \lambda_2^{(n)}, \delta^{(n)})$-SID codes such that
  \begin{equation}
    \lim_{n\rightarrow\infty}\lambda_1^{(n)}
    = \lim_{n\rightarrow\infty}\lambda_2^{(n)}
    = \lim_{n\rightarrow\infty}\delta^{(n)}
    = 0
    .
  \end{equation}
  The secret identification capacity $\capSID(W_{Y|X},W_{Z|X})$
  is defined as the supremum of all achievable rates. It was found that $\capSID$ is given by~\cite[Thm. 1]{AhlswedeZhang1995Newdirectionstheory}
  \begin{align}
    \label{dichotomy}
    \capSID(W_{Y|X},W_{Z|X})=
    \begin{dcases*}
      \capID(W_{Y|X}) & if $\capST(W_{Y|X},W_{Z|X})> 0$, \\
      0               & otherwise.
    \end{dcases*}
  \end{align}
\end{definition}

\begin{remark}
  In a marked contrast to transmission, decision sets for different messages
  are allowed to overlap in identification codes. This is essential to
  achieve the doubly exponential growth of the message set size in the blocklength $n$.
\end{remark}

\subsection{Codes for Approximation of Output Statistics}

In this subsection, we define approximation-of-output statistics (AOS) codes
which serve as  important building blocks in our codes. Similar codes have been
considered in~\cite{Wyner1975commoninformationtwo,HanVerdu1993Approximationtheoryoutput,SudanTyagiWatanabe20,WatanabeHayashi14,
HouKramer2013divergence}, and in \cite[Section~V.A]{SudanTyagiWatanabe20}.

\begin{definition}[Approximation-of-Output Statistics (AOS) Code] \label{AOScodeDef}
  A $(|\cK|,n\,|\,\delta^{(n)})$-AOS code
  for a DMC $W_{Y|X}$ and an i.i.d reference output distribution $Q_Y^n$
  consists of the  following:
  \begin{enumerate}
    \item An index set $\cK$, a block length $n$ and a distortion $\delta^{(n)} > 0$,
    \item an encoder that selects an index $K$ uniformly at random from $\cK$
   , and sends $x^n(K)$, and sends this over $n$ uses of the channel $W_{Y|X}$. The codebook $C_o = \tup{ X^n(1), \dots, X^n(\card{\mathcal{K} })}$ is a collection of the codewords that corresponds to each index $k\in\mathcal{K}$ where each codeword is generated randomly and independently using the a PMF $P_X$.
  \end{enumerate}
   We define the output distribution  $P_{Y^n}$
   \begin{align}
    P_{Y^n}(y^n)= P_{\mathcal{K}}^{\textbf{unif}}\circ P_{X^n|K}\circ W^{n}_{Y|X}(y^n) \nonumber \\
    = \frac{1}{|\mathcal{K}|}\sum_{k} W^n_{Y|X}(y^n|x^n(k))
   \end{align}

  where $P_{X^n|K}(a^n|k) = \ind{a^n = x^n(k)}$, for $a^n \in \cX^n$.
  The KL-divergence between the distributions $P_{Y^n}$ and $Q_Y^n$
  must be bounded by  \begin{equation}
    D(P_{Y^n} \| Q_Y^n) \le \delta^{(n)}
    .
  \end{equation}
  A rate $R_{\text{AOS}} \coloneqq \frac{1}{n}\log \card{\cK_{\text{AOS}}}$
  is achievable if there is a sequence of $(\card{ \cK_{\text{AOS}} }, n \,|\, \delta^{(n)})$ codes
  such that $\lim_{n \to \infty} \delta^{(n)} = 0$.
  The infimum in $\bbR$ of all achievable AOS code rates is denoted by $\mathbf{R^{\star}_{AOS}}(W_{Y|X},Q_Y)$. It was found that $\mathbf{R^{\star}_{AOS}}$ is   given by \cite{WatanabeHayashi14},\cite[Thm. V.1.]{SudanTyagiWatanabe20}
  \begin{equation}
    \mathbf{R^{\star}_{AOS}}(W_{Y|X},Q_Y)=\min_{P_{X}\in \Paos} I(X;Y),
  \end{equation}
  where $\Paos = \set{ P_X \in \cP(\cX) : P_X\circ W_{Y|X} = Q_Y }$.
  Note that $\Paos$ is empty if no $P_X$ can simulate $Q_Y$ via $W_{Y|X}$,
  and $\mathbf{R^\star_{\text{AOS}}}$ is undefined, then, because the infimum is taken in $\bbR$.
  This will be discussed in \cref{remark:undefinedRates}, after the definitions
  of some related codes.
\end{definition}

\subsection{Combining Reliability, Secrecy, and Approximation of Output Statistics}

The object of study in this paper are effectively secret identification (ESID) codes, which combine identification,
secrecy and stealth, i.e., approximation of output statistics, over the attacker's channel.
In the following, we define transmission codes with and without secrecy or stealth,
which we will use to construct ESID codes. The latter are defined in the end of this section.

\begin{definition}[AOS-Transmission (AOST) code] \label{AOSTcodeDef}
  A $(|\cM_{\text{AOST}}|, n \,|\, \lambda^{(n)}, \delta^{(n)})$-AOST code
  for a pair of DMCs $(W_{Y|X}, W_{Z|X})$ and an i.i.d reference output distribution
  $Q_Z^n$ consists of the following:
  \begin{enumerate}
    \item a message set $\cM_{\text{AOST}}$,
    \item a determinstic encoder
      $f_{\text{AOST}} : \mathcal{M}_{\text{AOST}} \to \mathcal{X}^n$,
    \item a decoder $\phi: \cY^n \to \cM_{\text{AOST}}$.
      The decision set for message $m$ is $\mathcal{G}_m=\{y^n: \phi(y^n)=m\}$.
  \end{enumerate}
  We define the output distribution  $P_{Z^n}$
   \begin{align}
    P_{Z^n}(z^n)= P_{\mathcal{M}}^{\textbf{unif}}\circ P_{X^n|M}\circ W^{n}_{Z|X}(z^n) \nonumber \\
    = \frac{1}{|\mathcal{M}|}\sum_{m} W^n_{Z|X}(z^n|x^n(m))
   \end{align}

  where $P_{X^n|M}(a^n|m) = \ind{a^n = x^n(m)}$, for $a^n \in \cX^n$.
  An AOST code must satisfy the following conditions:
  \begin{align}
      \max_{m\in\mathcal{M}_{\text{AOST}}}W_{Y|X}^n(\mathcal{G}^{c}_{m}|x^{n}(m))&\leq \lambda^{(n)}\\
     D(P_{Z^n} \,\|\, Q_Z^n) &\le \delta^{(n)}.
  \end{align}
  A rate $R_{\text{AOST}} \coloneqq \frac{1}{n}\log{|\mathcal{M}_{\text{AOST}}|}$ is
  called achievable if there is
  a sequence of $(|\mathcal{M}_{\text{AOST}}|, n\, |$ $\, \lambda^{(n)}, \delta^{(n)})$ codes such that
  \begin{equation}
    \lim_{n\rightarrow\infty}\lambda^{(n)}
    = \lim_{n\rightarrow\infty}\delta^{(n)}
    = 0
    .
  \end{equation}
  The AOST capacity $\capAOST(W_{Y|X},W_{Z|X},Q_Z)$ is the supremum in $\bbR$ of all achievable rates.
  It is subsequently given in \cref{AOST.Theorem}.
\end{definition}

\begin{definition}
  [{Effectively secret transmission (EST) code \cite{HouKramer14a,HouKramerBloch2017EffectiveSecrecyReliability}}]
  \label{ESTcodes.defn}
  A $(|\mathcal{M}_{\mathsf{EST}}|,n\,|\,\lambda^{(n)}, {\delta}^{(n)})$-EST code
  for a DM-WTC $(W_{Y|X},W_{Z|X})$,
  and an i.i.d reference distribution $Q_Z^n$ is a $(|\mathcal{M}_{T}|, n\,|$ $\,\lambda^{(n)})$-transmission code
  which additionally satisfies the condition\footnote{
    This semantic stealth constraint is equivalent to semantic effective
    secrecy, as explained in~\cref{remark:semanticStealthEquivSemanticEffectiveSecrecy}.
  }
  \begin{align}
   \max\limits_{m}D( E_m \circ W_{Y|X}^n\| \, Q_Z^n)\leq \delta^{(n)}.
  \end{align}
  The capacity $\capEST (W_{Y|X},W_{Z|X},Q_Z)$ is defined as the supremum in $\bbR$ of all achievable rates.
  It was found that $\capEST$ is given by~\cite{HouKramer14a,HouKramerBloch2017EffectiveSecrecyReliability}
  \begin{equation}
    \capEST(W_{Y|X},W_{Z|X},Q_Z)=\max_{P_{U,X}\in\Pest} I(U;Y)-I(U;Z) ,
  \end{equation}
  where $\Pest = \set{
    P_{U,X} \in \cP(\cU, \cX)
    : \card \cU \le \card \cX
    ,\, U - X - (Y, Z)
    ,\, P_X \circ W_{Y|X} = Q_Z
  }$.
\end{definition}

\begin{definition}[Effectively secret identification (ESID) code] \label{ESID.code.defn}
  A $(|\mathcal{M}_{\text{ESID}}|, n \,|\, \lambda_1^{(n)}, \lambda_2^{(n)}, \delta^{(n)})$-ESID code
  for a DM-WTC $(W_{Y|X}, W_{Z|X})$ and a reference output distribution $Q_Z^n$ is a
  $(|\mathcal{M}_{\text{ID}}|, n \,|\, \lambda_1^{(n)}, \lambda_2^{(n)})$-SID code
  that satisfies the condition\footnotemark[\thefootnote]
  \begin{align}
    \max\limits_{m}D(E_m\circ W_{Y|X}^n\| \, Q_Z^n)\leq \delta^{(n)}.
  \end{align}
  A rate $R_{\text{ESID}}\coloneqq\frac{1}{n}
  {\log\log|\mathcal{M}_{\text{ESID}}|}$ is
  called achievable if there exists an infinite sequence of
  $(|\mathcal{M}_{\text{ESID}}|, n\,| \lambda_1^{(n)},\lambda_2^{(n)},\delta^{(n)})$
  ESID codes that satisfies
  \begin{align}
    \lim_{n\rightarrow\infty}\lambda_1^{(n)}= \lim_{n\rightarrow\infty}\lambda_2^{(n)}=
    \lim_{n\rightarrow\infty}\delta^{(n)}=0.
  \end{align}
  The effectively-secrect-identification capacity
  $\capESID(W_{Y|X},W_{Z|X},Q_Z)$ is the supremum in $\bbR$ of all achievable rates.
\end{definition}

\begin{remark}
  \label{remark:undefinedRates}
  If $\Paos$, $\Paost$, or $\Pest$ are empty, i.e.
  no $P_X$ can simulate $Q_Z$ via the channel $W_{Z|X}$, then, equivalently, no rate is achievable,
  because for sufficiently small $\delta^{(n)}$, the KL-divergence of the output
  distribution $P_Z$ and $Q_Z$ is too large.
  This is why we explicitly emphasize that the parent set of all rates is assumed to be $\bbR$,
  which is our implicit assumption for all rates in this paper: it makes suprema and infima undefined
  if the set of achievable rates is empty, because every element of the parent set is an upper and lower bound
  on the empty set and the infimum (resp. supremum) of the empty set is the
  largest (resp. smallest) element of the parent set, in our case $\inf \emptyset = \max \bbR$
  and $\sup \emptyset = \min \bbR$, which both do not exist.
  Thus, the results mirror the intuition: if no rate can be achieved, not even zero, then
  the optimal rate is undefined.
  Thus, the capacity is only zero if the trivial code with only one message can be used to simulate the
  reference distribution.
  In contrast, for transmission and identification, with or without secrecy, rate zero is always achievable:
  For the trivial code, the error probability is always zero, as is the mutual information in the secrecy criterion.
  It seems more conventional to always allow rate zero in the definition of capacities,
  even for AOST, EST (as in \cite{HouKramer14a,HouKramerBloch2017EffectiveSecrecyReliability}) and ESID,
  by implicitly considering the supremal rates in the parent set of all non-negative reals, where
  $\sup \emptyset = \min [0, +\infty) = 0$.
  However, if the capacity were zero in both cases, when the set of achievable rates or only its interior is empty,
  the results could not express the difference between the possibility of a rate-zero transmission/identification
  which still seems innocent to the attacker, and the impossibility of seeming innocent even
  with a useless code without any actual communication.
\end{remark}

\section{Results}
\label{sec:results}

In this section, we present the main results of the paper.
Proofs are omitted here, and they are presented in \cref{sec:esid.dm.proof}.
We begin with a capacity theorem for simultaneous approximation of output
statistics over one branch of a wiretap channel and transmission over the other
(\cref{AOST.Theorem}).
This is a new result in its own right, and used here to prove the
subsequent lower bound for identification with effective secrecy in
\cref{prop:achiev.X}.

\begin{theorem}
\label{AOST.Theorem}
For any $Q_Z \in \cP(\cZ)$ and a pair of channels $(W_{Y|X},W_{Z|X})$
\begin{equation}
\capAOST(W_{Y|X},W_{Z|X},Q_Z)= \max_{P_{X}\in \Paost
 } I(X;Y) ,
\end{equation}
where $\Paost\coloneqq\{P_{X} \in \mathcal{P}_{X}:\, I(X;Y)\geq I(X;Z),\, P_X \circ W_{Z|X}=Q_Z \}$.
\end{theorem}

\begin{theorem}
  \label{prop:achiev.X}
  For any $Q_Z \in \cP(\cZ)$,
   $\capESID(W_{Y|X},W_{Z|X},Q_Z)$ satisfies
  \begin{gather}
  \label{eq.ESID.LowerBound}
    \capESID(W_{Y|X},W_{Z|X},Q_Z)
    \ge
      \max_{P_{X}\in \Paost
 }
      I(X; Y),
  \end{gather}
  where $\Paost$ is defined as in \cref{AOST.Theorem}.
\end{theorem}

In \cref{prop:achiev.X}, the constraint $I(X; Y) \ge I(X;Z)$
can be relaxed to $I(U; Y) \ge I(U; Z)$,
by applying~\cref{prop:achiev.X} to the virtual channel
$P_{YZ|U} = P_{X|U}\circ W_{YZ|X}$,
where $U$ can be any auxiliary random variable that
has finite support and satisfies the Markov condition $U - X - (Y,Z)$.
\begin{corollary}
  \label{corollary:achiev.U}
  $\capESID(W_{Y|X},W_{Z|X},Q_Z)$ satisfies:
  \begin{gather}
  \label{eq:ESID.LowerBound.U}
    \capESID(W_{Y|X},W_{Z|X},Q_Z)
    \ge
      \max_{
      P_{U, X} \in\Pesid}
      I(U; Y),
  \end{gather}
 where $\Pesid\coloneqq\{P_{U,X} \in \cP(\cU \times \cX):U - X - (Y,Z), \, I(U;Y)\geq I(U;Z),\, P_X \circ W_{Z|X}=Q_Z \}$,
 and $\cU$ is any finite set.
\end{corollary}

This lower bound may well also be an upper bound.
However, the proof of this remains elusive, as our following upper bound
has a different form.
\begin{theorem}
  \label{thm:esid.dm}
   $\capESID(W_{Y|X},W_{Z|X},Q_Z)$
  is $0$ if $I(U; Y) < I(U; Z)$ or $P_X \circ W_{Z|X} \ne Q_Z$, for all $P_{U,X}$ such that
  $U - X - (Y,Z)$. Otherwise, it
  satisfies:
  \begin{gather}
  \label{eq:esid.UpperBound}
    \capESID(W_{Y|X},W_{Z|X}, Q_Z) \le
    \max_{
      P_{U, X} \in\Pesid}
    I(X; Y),
\end{gather}
  where $\Pesid$ is defined as in \cref{corollary:achiev.U},
  and $\card\cU \le \card\cX$.
\end{theorem}

On one hand, this upper bound differs from the lower bound in
\cref{corollary:achiev.U} in the first argument of the mutual information:
the lower bound has $I(U; Y)$ and the upper bound $I(X; Y)$.
One the other hand, while these bounds do not coincide in general, they do so in special cases,
the most interesting of which is, perhaps, the class of more capable channels.

\begin{corollary}
  \label{corollary:esid.dm.moreCapable}
  If $W_{Y|X}$ is more capable than $W_{Z|X}$,
  \begin{gather}
    \capESID(W_{Y|X},W_{Z|X}, Q_Z) =
    \max_{
      P_X\in \Paos}
    I(X; Y),
  \end{gather}
  where $\Paos = \Paost = \set{P_X \in \cP(\cX) : P_X \circ W_{Z|X} = Q_Z }$ refers to an AOS code for the channel $W_{Z|X}$
  (see \cref{AOScodeDef}).
\end{corollary}
\begin{proof}
For more capable channels, $I(X; Y) \ge I(X; Z)$, for all $P_X \in \cP(\cX)$
and hence, $\Paost = \Paos$.
Since $\set{ P_X : \exists\, P_{U,X} \in \Pesid} \subseteq \Paos$, the lower bound in \cref{prop:achiev.X} is at
least as large as the upper bound in \cref{thm:esid.dm},
and hence the bounds are equal.
\end{proof}

\begin{remark}
  The lower bound in \cref{corollary:achiev.U}
  and the upper bound in \cref{thm:esid.dm} coincide only
  for channels where the capacity is achieved
  with $U = X$.
  In \cref{sec:example}, we discuss several examples of more capable
  channels, as well as a pair of reversely degraded channels,
  which is not more capable,
  and for certain channel parameters there is a gap between the
  upper and lower capacity bounds.
  On the other hand, there are parameters for which
  they coincide indeed; thus, the capacity is determined exactly.
\end{remark}
\begin{remark}
  If $\capESID(W_{Y|X},W_{Z|X},Q_Z)$ is zero,
  then effectively secret communication with any positive rate is
  impossible.
  This does not necessarily imply that communication is
  impossible. It can simply mean that the code size grows slower than
  doubly-exponentially in the block length, since we defined the
  rate as $R_{\text{ESID}} = \frac{1}{n} \log\log |\mathcal{M}_{\text{ESID}}|$.
    For example, for covert codes $O(\sqrt n)$ bits can be sent in $n$ channel uses.  However,
  Ahlswede~\cite[Lemmas 4, 5 and Remark 4]{AhlswedeZhang1995Newdirectionstheory}
  proved that if for sufficiently small $\lambda_1^{(n)},\lambda_2^{(n)},\delta^{(n)}$
  and sufficiently large $n$,
  the secrecy condition $I(U; Y) \geq I(U; Z)$ is
  violated for all $U$, then secret communication is impossible,
  hence effectively secret communication as well.
\end{remark}
\begin{remark}
\label{BoundsDefn}
For the rest of the paper
the lower bound in \cref{eq.ESID.LowerBound} is denoted by
$R_{\text{L}}$, the lower bound in \cref{eq:ESID.LowerBound.U} is denoted by $R_{\text{L}, U= X}$, and the upper bound in \cref{eq:esid.UpperBound} is denoted by  $R_{\text{U}}$.
\end{remark}

\section{Examples}
\label{sec:example}

In the following examples, consider $W_{Y|X}$ to be the marginal point-to-point channel
between Alice and Bob, and $W_{Z|X}$ to be the marginal point-to-point channel between
Alice and Willie. The reference distribution at Willie is an i.i.d. distribution $Q_Z^n$.
Whenever capacities in this section refer to this assignment,
we drop the arguments and simply write $\capSID, \capST, \capESID, \dots$.

\subsection{Degenerate Channels}

\begin{figure}
  \centering
  \includegraphics[scale=0.6]{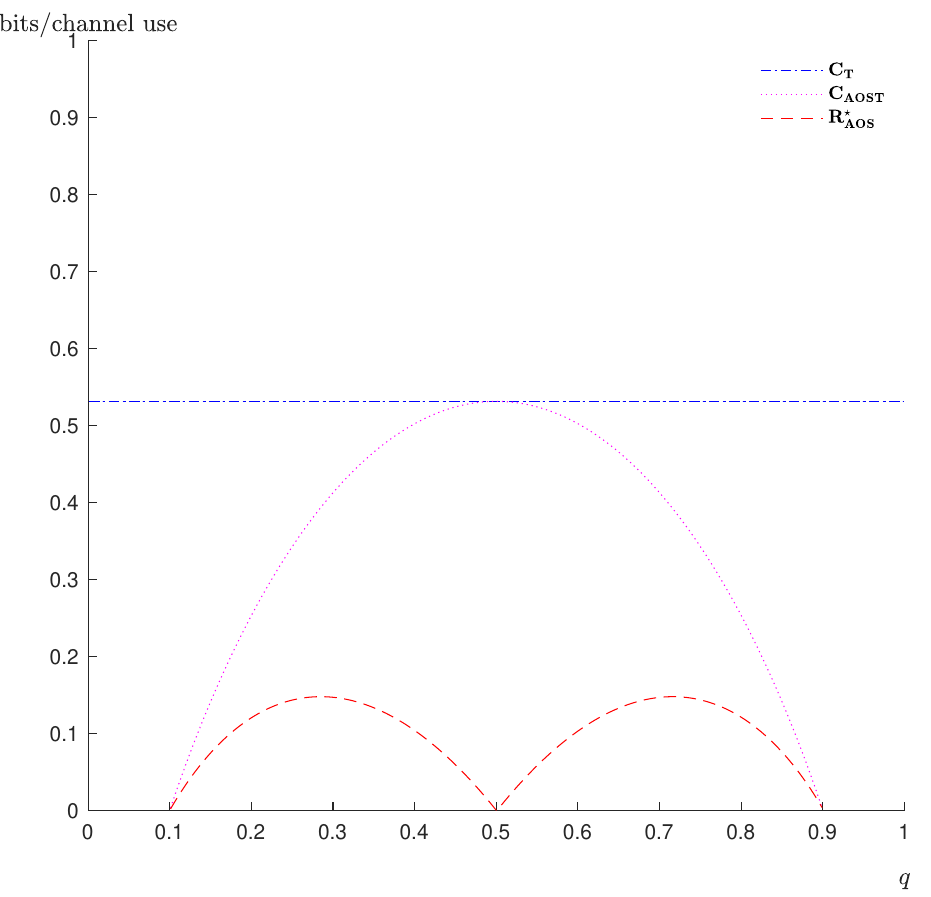}
  \caption{\label{fig:NonFullRank}
  $\capT$, $\capAOST$ and $\textbf{R}^{\star}_{\text{AOS}}$
  of the DM-WTC in \cref{NonFullRankExample} with $\epsilon=0.1$,
  for different values of $q$.
  Note that both $\textbf{R}^{\star}_{\text{AOS}}$
  and $\capAOST$ are undefined for
  $q \in [0,0.1) \cup (0.9,1]$.}
\end{figure}

In this subsection, we consider an example of channels where the consideration of simulating a uniform reference output channel behavior with the least amount of a uniform input randomness can be achieved by sending a particular symbol over the channel, i.e., there is no need to randomize over the channel. While the need to maximize the number of information bits over the channel would entail using a uniform input distribution over a subset of the input symbols which effectively render the channel as a BSC. If the reference output distribution is not uniform, simulating the reference output distribution begins to shift away from using only one fixed symbol and other symbols of the input alphabet have to be used in order to be able to simulate the reference output distribution. While the input distribution needed to both maximize the number of information bits received by Bob and simulate the reference output distribution at Willie has to shift away from the unifrom distribution.

Consider the input alphabet $\cX = \set{ 0, e , 1 }$, output alphabets
$\cY = \cZ = \set{ 0,1 }$, and the channel transition matrices
\begin{gather}
  \label{NonFullRankExample}
  W_{Y|X}= W_{Z|X} =
  \tup{
    \begin{smallmatrix}
      1 - \epsilon \,  &\, \epsilon \\
      1/2 \,  &\, 1/2 \\
      \epsilon \, & \, 1 - \epsilon
    \end{smallmatrix}
  }
  \,.
\end{gather}
We denote $q \coloneqq Q_Z(1)$ and $p \coloneqq P_X(1)$.
It can be noted that although the symbol $e$ is useless for communication,
for certain values of $q$, it can be used to simulate the reference output distribution, by inputting $e$ for a fraction of input symbols instead of communicating,
and relying on the channel to randomize.
Since for pure communication, the channels $W_{Y|X}$ and $W_{Z|X}$
are effectively $\BSC(\epsilon)$, their transmission capacities are given by
\begin{align}
  \capT(W_{Y|X}) = \capT(W_{Z|X}) = 1 - H_b(\epsilon)
  \,.
\end{align}
Furthermore, for $q = \frac 1 2$, the reference distribution can be simulated
by constantly inputting $e$. This strategy is optimal, i.e. it minimizes the randomness in the channel input used to simulate $Q_Z$,
and thus
\begin{align}
  \textbf{R}^{\star}_{\text{AOS}}(W_{Z|X},Q_Z)
  &= \min_{P_X\in \Paos} I(X;Z) \nonumber \\
  &= \min_{P_X\in \Paos} H(Z) - H(Z|X) = 1 - 1 = 0
  \,.
\end{align}
In general, $\textbf{R}^{\star}_{\text{AOS}}(W_{Z|X}, Q_Z)$ is a function of $q$,
and $P_X$ exactly simulates $Q_Z$ over $W_{Z|X}$ if and only if
\begin{align}
  q
  &= P_Z(1) \nonumber \\
  &= P_X(1) W_{Z|X}(1|1) + P_X(e) W_{Z|X}(1|e) + P_X(0) W_{Z|X}(1|0) \nonumber \\
  &= p\, (1 - \epsilon) + \frac{p_e}{2} + (1 - p - p_e)  \epsilon
  \,.
\end{align}
Equivalently, we can write
\begin{equation}
2 p + p_e = \frac{q-\epsilon}{\frac{1}{2}-\epsilon}
\,.
\label{SimulationCondition1}
\end{equation}
Since $0 \le \min \set{ p, p_e } \le p + p_e \le 1$ and thus
$0 \le \frac{q-\epsilon}{1/2 - \epsilon} \le 2$,
it follows that
\begin{gather}
  \min\set{ \epsilon, 1 - \epsilon } \le q \le \max\set{ \epsilon, 1 - \epsilon }
  \,.
  \label{SimulationConditionExistential}
\end{gather}
This is a necessary and sufficient condition for the existence of a $P_X$ that simulates $Q_Z$ over $W_{Z|X}$.
If the condition is violated, e.g. if $q < \epsilon < \frac{1}{2}$,
then $\textbf{R}^{\star}_{\text{AOS}}(W_{Z|X})$ is undefined,
as simulation of $Q_Z$ is impossible with any rate $0 \le R \le \infty$.
Otherwise, for any $q$ satisfying \cref{SimulationConditionExistential},
there exist $p, p_e$ satsifying

\begin{align}
  \textbf{R}^{\star}_{\text{AOS}}(W_{Z|X}, Q_Z)
  &= \min_{P_X \in \Paos} I(X;Z) \nonumber \\
  &= \min_{P_X \in \Paos} H(Z) - H(Z|X) \nonumber \\
  &= H_b(q) - p H(Z|X=1) - p_e H(Z|X=e)
  \twocol{\nonumber \\ &\quad}
     - (1 - p - p_e) H(Z|X=0) \nonumber \\
  &= H_b(q) - (1 - p_e) H_b(\epsilon) - p_e H_b(\frac{1}{2}) \nonumber \\
  &= H_b(q) - H_b(\epsilon) - p_e(1 - H_b(\epsilon)) \nonumber \\
  &= H_b(q) - H_b(\epsilon)
    - \tup{ \frac{q - \epsilon}{\frac{1}{2} - \epsilon} - 2p }
      (1 - H_b(\epsilon))
    \nonumber \\
  &= H_b(q) - H_b(\epsilon) - \frac{q - \epsilon}{\frac{1}{2} - \epsilon}(1 - H_b(\epsilon))
  \,,
\end{align}
where the last equality holds since $p$ can be chosen to minimize the mutual information, and the left-hand side expression
is minimized by $p = 0$.

For AOST, we can derive an expression for the capacity in a similar way,
but maximizing the mutual information.
Hence, for every $q$ satisfying \cref{SimulationConditionExistential},
there exists $\tilde p \in [0,1]$ such that
\begin{align}
  \capAOST
  &= \max_{P_X \in \Paost} I(X; Y) \nonumber \\
  &= \max_{P_X \in \Paost} H(Y) - H(Y|X) \nonumber \\
  &= H_b(q) - \tilde p H(Z|X=1) - \tilde p_e H(Z|X=e)
  \twocol{\nonumber \\ &\quad}
     - (1 - \tilde p - \tilde p_e) H(Z|X=0) \nonumber \\
  & = H_b (q) - H_b(\epsilon)
  \,,
\end{align}
where $\tilde p_e = 0$.
Since AOST involves AOS, if \cref{SimulationConditionExistential}
is violated, $\capAOST$ is also undefined.
In \cref{fig:NonFullRank}, $\textbf{R}^{\star}_{\text{AOS}}(W_{Z|X})$ is visualized for $\epsilon=0.1$ and different values of $q$. It is contrasted with $\capAOST(W_{Y|X},W_{Z|X})$ as well as $\capT(W_{Y|X})$. I We note it is almost trivial to simulate the output reference around $q=\frac{1}{2}$, and that $\textbf{R}^{\star}$ is undefined as $q\in[0,\epsilon]\cup[1-\epsilon, 1]$.

\subsection{More Capable Channels}

Here, we present three examples of more capable DM-WTCs, where $\capESID$ is fully determined. We contrast the behavior of $\capESID$ with the channel capacities of the same DM-WTC for several communications tasks.

\subsubsection{Two Binary Symmetric Channels}
\label{TwoBSC}

\begin{figure}
  \centering
  \includegraphics[scale=0.6]{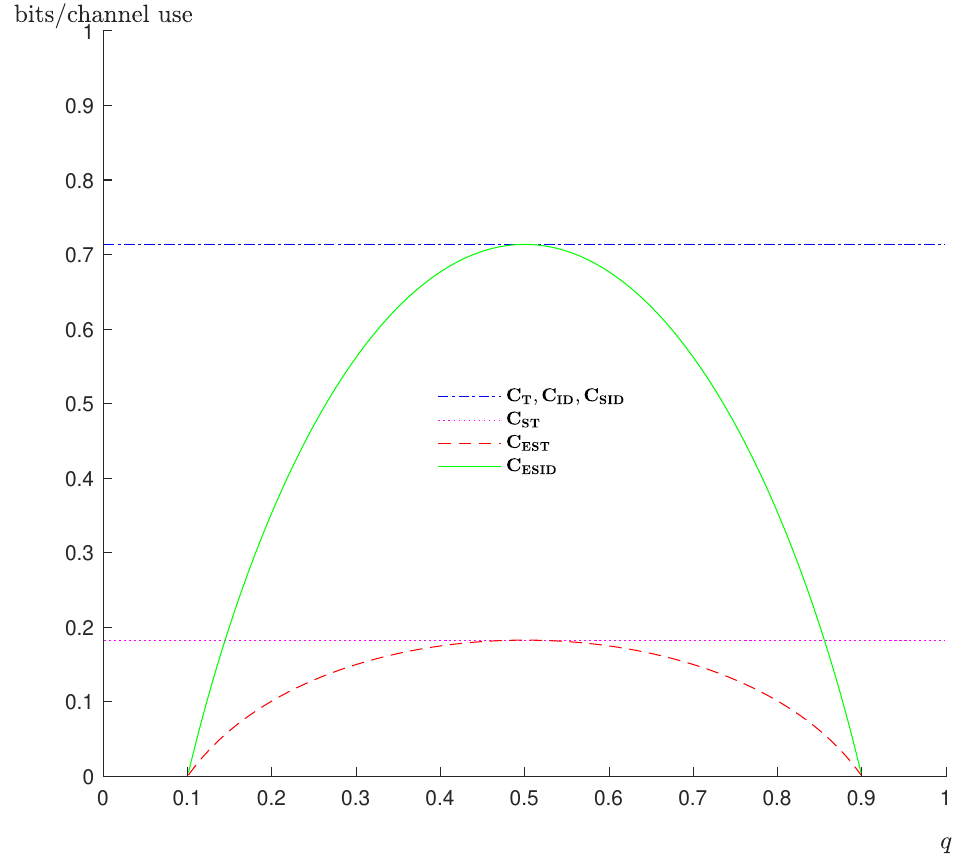}
  \caption{\label{fig:BSCs}
    A comparison between different capacities behavior with changing $q$ for DM-WTC in \cref{TwoBSC} with $\epsilon_1=0.05$ and $\epsilon_2=0.1$.}
\end{figure}

For arbitrary $\epsilon_1, \epsilon_2 \in [0,\frac 1 2]$,
consider a DM-WTC $(W_{Y|X}, W_{Z|X}) = (\BSC(\epsilon_1), \BSC(\epsilon_2))$.
Hence, the input and output alphabets are $\cX = \cY = \cZ = \set{0,1}$.
As above, let $q \coloneqq Q_Z(1)$ and $p \coloneqq P_X(1)$.
We restrict our attention to the case where $\epsilon_1 \leq \epsilon_2$,
i.e., more capable DM-WTCs.
Without secrecy constraints, the identification capacity is given by
\begin{align}
  \capID(W_{Y|X})=\capT(W_{Y|X})
  &= \max_{P_X \in \cP(\cX)} I(X;Y)
   \twocol{\nonumber \\ &}
    = 1-H_b(\epsilon_1)
  \,.
\end{align}
In comparison, the secret-transmission capacity is given by
\begin{align}
  \capST
  &= \max_{P_X \in \cP(\cX)} I(X; Y) - I(X; Z) \nonumber \\
  &=\max_p [H_b(p \star \epsilon_1) - H_b(\epsilon_1)] - [H_b(p \star \epsilon_2) - H_b(\epsilon_2)] \nonumber \\
  & = [1 - H_b(\epsilon_1)] - [1 - H_b(\epsilon_2)] \nonumber \\
  &= H_b(\epsilon_2) - H_b(\epsilon_1)
  \,,
\shortintertext{and the secret identification capacity is}
  \capSID
  &= \max_{P_X \in \mathcal{P}(\mathcal{X}) } I(X;Y)
  \twocol{\nonumber \\ &}
   = 1 - H_b(\epsilon_1)
  \,.
\end{align}

Regarding $\capAOST, \capEST$ and $\capESID$, it is essential to distinguish between two cases.
First, if $q < \epsilon_2$ or $q > 1 - \epsilon_2$,
then the reference distribution $Q_Z$ cannot be simulated, as
\begin{align}
  \set{ p : p \in [0,1],\, q = p \star \epsilon_2 } = \emptyset
  \,,
\end{align}
where we use the notation $p \star \epsilon_2 \coloneqq p (1-\epsilon_2) + \epsilon_2 (1-p)$.
Accordingly,
$\mathbf{R^\star_{\text{AOS}}} (W_{Z|X})$, $\capAOST$, $\capEST$ and $\capESID$ are undefined.
Second, if $\epsilon_2 < q < 1 - \epsilon_2$,
the reference distribution $Q_Z$ can be simulated,
using some $p$ that is uniquely determined by $q$.
Therefore,

\begin{align}
  \capEST
  &= \max_{ P_{X}\in\Pest} I(X ; Y) - I(X; Z) \nonumber \\
  &= \brack{ H_b(p\star\epsilon_2) - H(\epsilon_1) } - \brack{ H_b(q)-H_b(\epsilon_2) }.
\end{align}
Furthermore,
\begin{align}
  \capAOST = \capESID
  &= \max_{P_X \in \Pest} I(X;Y) \nonumber \\
  &= H_b(p\star\epsilon_2)-H(\epsilon_1)
  \,.
\end{align}
In contrast, it holds that $\mathbf{R^\star_{AOS}} (W_{Z|X}) = H_b(q) - H_b(\epsilon_2)$.
Notice that if $W_{Z|X}$ is strictly more capable than $W_{Y|X}$,
  i.e., if $\epsilon_1 > \epsilon_2$,
  then $\capESID$, $\capEST$ and $\capAOST$ do not exist,
  since $\Paost$
\begin{align}
  \capESID = \capEST =\capAOST = \capSID = \capST = 0 \,.
\end{align}

\subsubsection{Two Binary Erasure Channels}
\label{TwoBEC}

\begin{figure}
  \centering
  \includegraphics[scale=0.6]{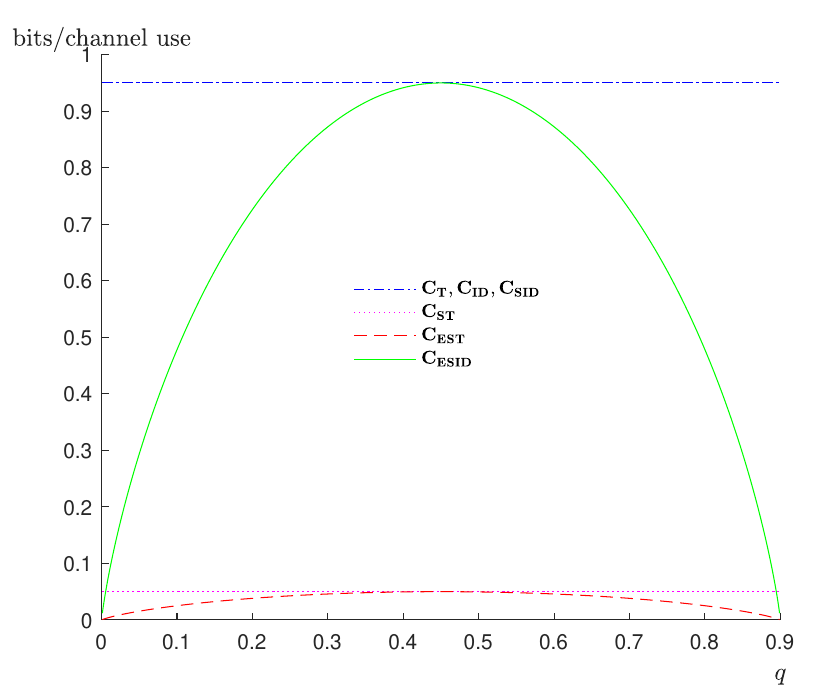}
  \caption{\label{fig:BECs}
    A comparison between different capacity functions with changing $q$ for DM-WTC in \cref{TwoBEC} with $\epsilon_1=0.05$ and $\epsilon_2=0.1$.}
\end{figure}

Consider a DM-WTC $(W_{Y|X}, W_{Z|X})$ where $W_{Y|X}=\BSC(\epsilon_1)$ and $W_{Z|X}=\BEC(\epsilon_2)$. The input alphabet is $\mathcal{X}=\{0,1\}$, and the output alphabets $\mathcal{Y}=\mathcal{Z}=\{0,1,e\}$. $W_{Y|X}=\BEC(\epsilon_1)$. Furthermore, $W_{Z|X}=\text{BEC}(\epsilon_2)$, where $\epsilon_1, \epsilon_2\in[0,1]$.
We denote $Q_Z(1)=q,\, P_X(1)=p$, where $p\in[0,1]$ and $q\in[0,1-\epsilon_2]$
as entailed by the channel. Without secrecy constraints, transmission and
identification capacities are given by
\begin{align}
\capID(W_{Y|X})=\capT(W_{Y|X}) &=\max_{P_X \in \mathcal{P}_\mathcal{X}}I(X;Y) \nonumber \\
&=1-\epsilon_1,
\end{align}
for all possible combinations of $\epsilon_1,\epsilon_2$ and $q$.
With secrecy constraints, the secret transmission capacity is calculated as follows:
\begin{align}
\capST & = \max_{P_X \in \mathcal{P} (\cX)} I(X;Y) - I(X;Z) \nonumber \\
 &=\epsilon_2 - \epsilon_1 .
\end{align}
Calculating $\capSID$, we notice that it is positive as $\capST$ is positive, and
\begin{align}
\capSID &= \max_{P_X \in \mathcal{P}(\mathcal{X})}\, I(X;Y) \nonumber \\
&=1-\epsilon_1.
\end{align}

Regarding $\mathbf{R}_{\text{AOS}}^{\star},  \capAOST, \capEST$ and $\capESID$, in contrast to the previous example of two BSCs, all possible output distributions that correspond with
$q\in[0:1-\epsilon_2]$ could be simulated here. Now,
\begin{align}
\mathbf{R}_{\text{AOS}}^{\star}(W_{Z|X}) & = \min_{P_X\in \Paos} I(X ; Z) \nonumber \\
& = H(Z) - H(Z|X) \nonumber \\
& = H( p (1 - \epsilon_2) , (1-p) (1 - \epsilon_2), \epsilon_2) - H_b (\epsilon_2) \nonumber \\
& = [H_b \big(\frac{q}{1 - \epsilon_2} \big)(1 - \epsilon_2 ) + H_b(\epsilon_2)]
- H_b(\epsilon_2) \nonumber \\
& = H_b \big(\frac{q}{1 - \epsilon_2} \big) (1 - \epsilon_2 ).
\end{align}
For $\capEST$ it holds
\begin{align}
\capEST &=\max_{P_{X}\in\Paos}{I(X;Y)-I(X;Z)} \nonumber \\
&= H_b(\frac{q}{1-\epsilon_2})(1-\epsilon_1)-H_b(\frac{q}{1-\epsilon_2})(1-\epsilon_2) \nonumber \\
&= H_b(\frac{q}{1-\epsilon_2})(\epsilon_2-\epsilon_1).
\end{align}
Furthermore, for $\capAOST$ and $\capESID$ we have
\begin{align}
\capAOST= \capESID&=\max_{P_{X}\in\Paos}I(X;Y) \nonumber \\
&=H_b\big(\frac{q}{1-\epsilon_2}\big)(1-\epsilon_1).
\end{align}

Notice that if $(W_{Y|X}, W_{Z|X})$ is not more capable DM-WTC, i.e., $\epsilon_2 \geq \epsilon_1$, then
\begin{align}
\capESID = \capEST =\capAOST =\capSID = \capST =  0 .
\end{align}

\subsubsection{A Binary Erasure Channel and a Binary Symmetric One}
\label{MixedExample}

\begin{figure}
  \centering
  \includegraphics[scale=0.6]{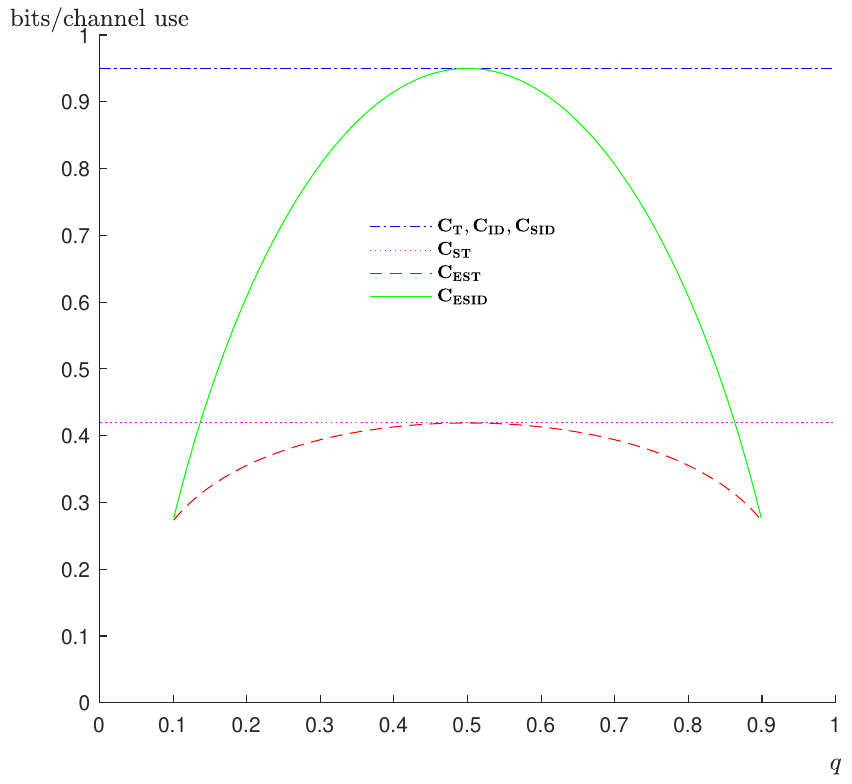}
  \caption{\label{fig:BEC+BSC}
    A comparison between different capacities behavior with changing $q$ for DM-WTC in \cref{MixedExample} with $\epsilon_1=0.05$ and $\epsilon_2=0.1$. Note that $\capEST(q)$ and $\capESID(q)$ are undefined over the intervals $q\in [0,0.1)$ and $q\in(0.9,1]$. }
\end{figure}

For $\epsilon_1\in[0,1]$ and $\epsilon_2\in[0,\frac 1 2]$,
consider a DM-WTC $(W_{Y|X}, W_{Z|X})$,
where $W_{Y|X}=\text{BEC}(\epsilon_1)$ and $W_{Z|X}=\text{BSC}(\epsilon_2)$.
The input alphabet is $\mathcal{X}=\{0,1\}$, and the output alphabets are
$\mathcal{Y}=\{0,1,e\}$ and $\mathcal{Z}=\{0,1\}$.
We denote $P_X(1)=p,\, Q_Z(1)=q$,\, where $q\in[0,1]$ as entailed by the channel structure.
We restrict our attention to the case $H_b(\epsilon_2) > \epsilon_1$,
where the DM-WTC is strictly more capable \cite[p. 87]{BlochBarros2011phys_sec_book}.
It holds that
\begin{align}
  \capID(W_{Y|X})=\capT(W_{Y|X})&= \max_{P_X \in \cP(\cX)} I(X; Y) \nonumber \\
  &=1-\epsilon_1.
\end{align}
Considering the secrecy capacities we have
\begin{align}
  \capST & = \max_{P_X \in \cP(\cX)} I(X; Y) - I(X ; Z) \nonumber \\
  & = \max_p \brack{ H_b(p)(1-\epsilon_1)-(H_b(p\star\epsilon_2)-H_b(\epsilon_2)) } \nonumber \\
  & = H_b(\epsilon_2) - \epsilon_1
  \,,
\end{align}
where $\capST$ is strictly positive for more capable DM-WTC. Furthermore,
\begin{align}
  \capSID & = \max_{P_X \in \mathcal{P}(\cX)} I(X ; Y) \nonumber \\
  & = 1-\epsilon_1 .
\end{align}
Calculating $\capEST$ we get,
\begin{align}
  \capEST &=\max_{P_{X}\in\Paos}{I(X;Y)-I(X;Z)} \nonumber \\
  &= H_b\big(\frac{q}{1-\epsilon_2}\big)\big(1-\epsilon_1\big) - \big(H_b(q) - H(\epsilon_2)\big).
\end{align}
Furthermore, we get for $\capAOST$ and $\capESID$,
\begin{align}
  \capAOST= \capESID&=\max_{P_{X}\in\Paos}I(X;Y) \nonumber \\
  &=H_b\big(\frac{q}{1-\epsilon_2}\big)\big(1-\epsilon_1\big).
\end{align}
\begin{remark}
  \cref{fig:BEC+BSC} shows that at the cutoff values of
  $q=\{\epsilon_1,1-\epsilon_1\}$, the capacities
  $\capESID$ and $\capEST$ are strictly positive.
  This is in contrast to \cref{TwoBEC}, where \cref{fig:BSCs} shows
  that $\capESID=\capEST=0$ at the cutoff values of
  $q=\{\epsilon_1,1-\epsilon_1\}$.
\end{remark}

\subsection{Two Parallel, Reversely Degraded Subchannels}
\label{ReverseDegExamp}

\begin{figure}
    \centering
    \exampleChannelESID[xscale=1.1,yscale=0.5][1]
    \caption{\label{fig:revDegraded}
      Structure of DM-WTC in \cref{ReverseDegExamp} example.
      }
\end{figure}

\begin{figure}
  \centering
  \includegraphics[scale=0.6]{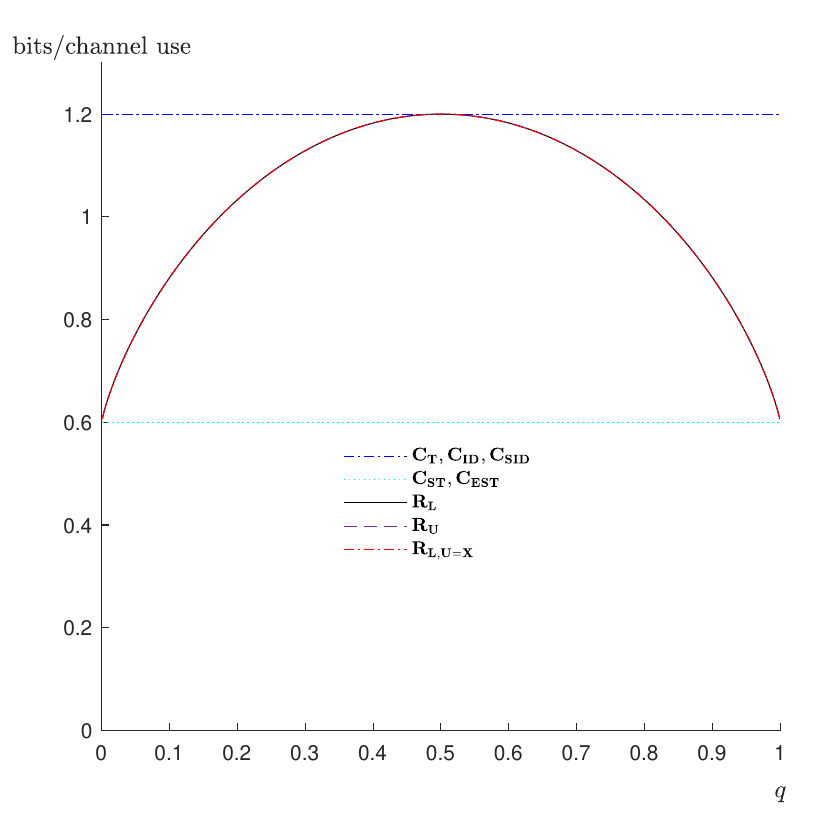}
  \caption{\label{fig:ReverseA}
    A comparison between different transmission and identification capacities as well as lower and upper bounds on $\capESID$  for $\epsilon=0.4$ for \cref{ReverseDegExamp}}
\end{figure}
\begin{figure}
  \centering
  \includegraphics[scale=0.6]{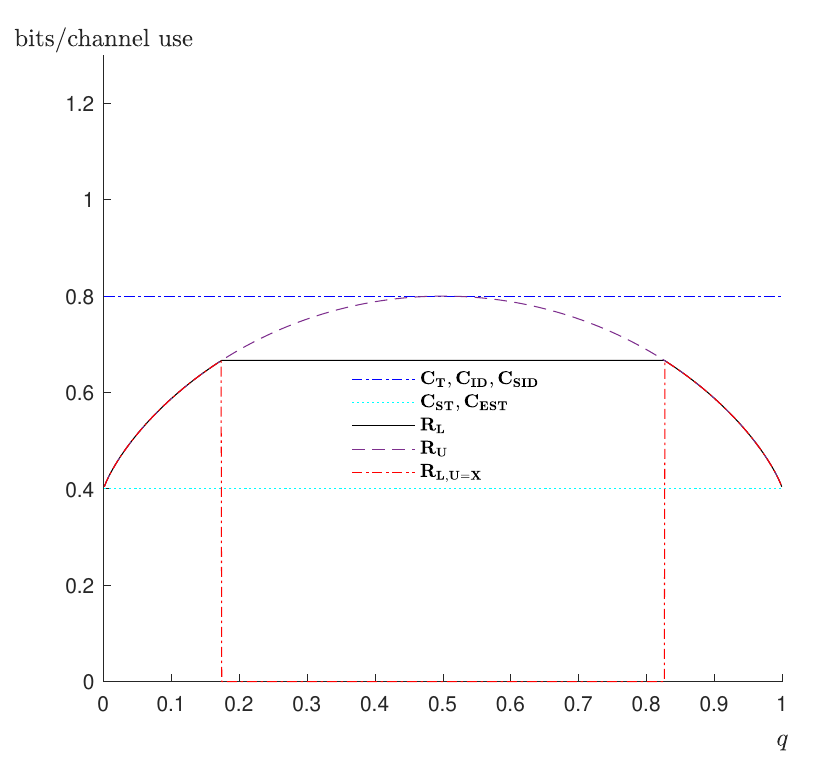}
  \caption{\label{fig:ReverseB}
    A comparison between different transmission and identification Capacities as well as lower and upper bounds on $\capESID$  for $\epsilon=0.6$ for \cref{ReverseDegExamp}. Note that for $\{q: H_b(q)\leq \frac{1}{\epsilon}\}$ randomization at the encoder is needed for attaining a positive effectively secure identification rate. And for our scheme, even with randomization, the achievable rate saturates as the value of $q$ changes.}
\end{figure}
\begin{figure}
  \centering
  \includegraphics[scale=0.6]{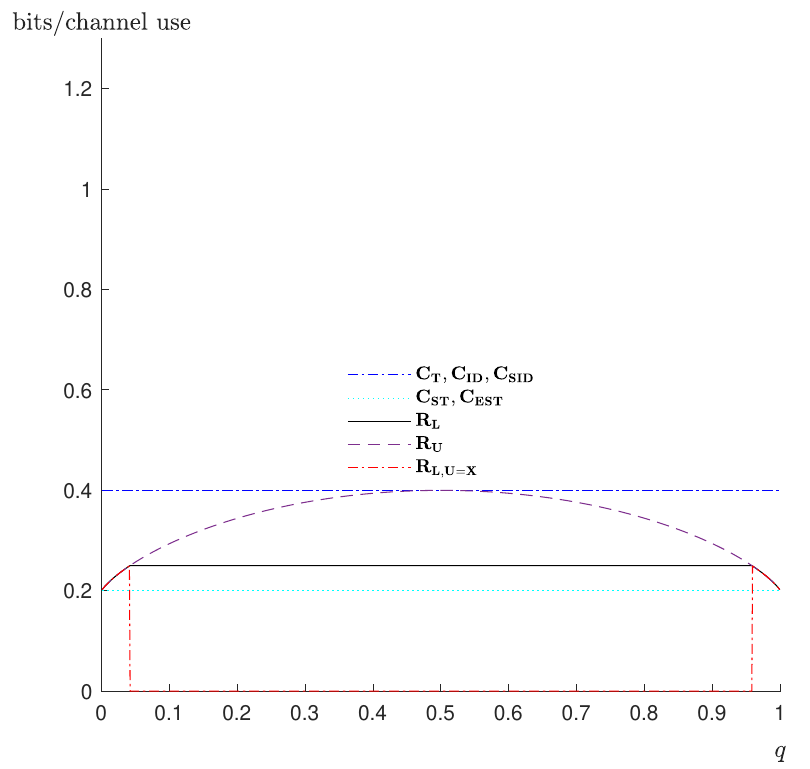}
  \caption{\label{fig:ReverseC}
    A comparison between different transmission and identification Capacities as well as lower and upper bounds on $\capESID$  for $\epsilon=0.8$ for \cref{ReverseDegExamp}. The same note regarding the necessity of randomization applies as in \cref{fig:ReverseB}. }
\end{figure}

In this example, we demonstrate the gap between the lower and upper
bounds on $\capESID$ as stated in
\cref{prop:achiev.X}, \cref{corollary:achiev.U} and \cref{thm:esid.dm}.
Furthermore, we contrast the behavior of these bounds with the channel
capacities of the same DM-WTC for several communication tasks.
Consider the following DM-WTC with two parallel subchannels,
$W_{Y_1|X_1} = \BEC(\epsilon)$ and $W_{Y_2|X_2} = \BEC(\epsilon)$,
as shown in \cref{fig:revDegraded}.
The input alphabet at Alice is
$\cX = \cX_1 \times \cX_2 = \set{ 0,1 } \times \set{ 0,1 }$,
the output alphabet at Bob is
$\cY = \cY_1 \times \cY_2 = \set{ 0,e,1 } \times \set{ 0,e,1 }$,
and the binary output alphabet at Willie is $\cZ = \set{ 0,1 }$.
Alice sends $X = (X_1, X_2)$, where each of $X_1$ and $X_2$ is sent over the
respective subchannel,
and Bob observes both channel outputs, i.e., he receives $Y=(Y_1,Y_2)$.
On the other hand, Willie receives $Z = X_2$, i.e.,
his observation of $X_2$ is noiseless but he has no information about $X_1$.
Equivalently, one could set $Z = (Z_e, X_2)$,
where $Z_e$ is a completely noisy version of $Y_1$ as sent over
$W_{Z_e|Y_1}=\BEC(1)$.
In this model, it is clear that $W_{Y_1,Z_e| X_1}$ is a stochastically degraded channel,
where $Z_e$ is a degraded version of $Y_1$.
Likewise, $W_{Y_2,Z| X_2}$ is a stochastically degraded channel as well,
but here $Y_2$ is a degraded version of $Z = X_2$.
Hence, it the two subchannels are reversely degraded with respect to each other.

In the following, we denote $P_{X_1}(1)=p_1,\, P_{X_2}(1)=p_2\, , Q_Z(1)=q$. Calculating $\capID (W_{Y|X})$,
$\capT(W_{Y|X})$, we get
\begin{align}
  \capID=\capT
  &=\max_{ P_X \in \cP(\cX)} I(X;Y) \nonumber \\
    & = \max_{P_{X_1}\in \mathcal{P}(\cX_1)} I(X_1;Y_1) + \max_{P_{X_2} \in \cP(\cX_2)} I(X_2;Y_2) \nonumber \\
  & = \max_{p_1}\, H_b(p_1)(1-\epsilon)+\max_{p_2}H_b(p_2)(1-\epsilon) \nonumber
  \\& = 2\, (1-\epsilon).
\end{align}

Computing $\capST$ we get
\begin{align}
\capST=&\max_{P_{U,X}\in\Pst} I(U;Y)- I(U;Z) \nonumber \\
\overset{(a)}{=}&\max_{P_{X_1} \in \cP(\cX_1)} I(X_1;Y_1)+\max_{P_{U_2,X_2} \in \mathcal{P}_{\mathsf{ST},2}}\brack{I(U_2;Y_2)-I(U_2;Z)} \nonumber  \\
\overset{(b)}{=}& \, (1-\epsilon)+ \max_{p_2, \delta} \big[ ( H_b(p_2\star\delta)-H_b(\delta))(1-\epsilon)-(H_b(p_2\star\delta)-H_b(\delta) \big] \nonumber  \\
\overset{(c)}{=} & \, 1-\epsilon,
\end{align}
where (a) holds by defining $\mathcal{P}_{\mathsf{ST}, 2} \coloneqq\{P_{U_2,X_2}\in\mathcal{P}_{U_2, X_2}: U - X -(Y, Z)\}$; (b) holds by introducing the possibility of randomization through preprocessing via
 a $\BSC(\delta)$, for $\delta \in [0,1]$; and (c) holds as there is no choice of $p_2, \delta \in [0,1]$ that renders the difference between the two terms in the maximization problem to be positive. Thus, $p_2 = \delta = 0$ could be chosen.

We compute $\capSID$ using \ref{dichotomy},
\begin{align}
\capSID=\capID &= \max_{P_X \in \mathcal{P}_{\text{SID}}} I(X; Y) \nonumber \\
&=2\,(1-\epsilon) .
\end{align}

Computing $\capEST$ we get
\begin{align}
\capEST & =\max_{P_{U,X}\in\Pest} I(U;Y)- I(U;Z) \nonumber \\
& \overset{(a)} {=}\max_{P_{X_1} \in \mathcal{P}(\cX_1)} I(X_1;Y_1)+\max_{P_{U_2,X_2 \in \mathcal{P}_{\mathsf{EST},2}}}\brack{I(U_2;Y_2)-I(U_2;Z)} \nonumber  \\
  & \overset{(b)} {=} (1-\epsilon)+ \max_{\delta} \, \large[(H_b(q)-H_b(\delta))(1-\epsilon)- ( H_b(q)-H_b(\delta))\large]
  \nonumber \\
  & \overset{(c)}{=} 1-\epsilon \,,
\end{align}
where (a) holds by defining
\[\cP_{\mathsf{EST}, 2} \coloneqq \set{
  P_{U_2,X_2} \in \cP(\cU_2, \cX_2) :
  U - X - (Y, Z),
  P_{X_2} \circ W_{Y_2|X_2} = Q_Z
},\]
as $q$ is entailed by the output reference distribution; (b) holds as the
choices of $\delta=q$ and $U\sim \text{Bern}(0)$ achieve capacity;
and (c) follows by choosing $\delta = q$.

Tackling the lower and upper bounds, we have on $\capESID$. We have to check the implications of the condition $I(U;Y)\geq I(U;Z)$ for a probability distribution $P_{U,X}$ to be within the feasible set $\Pesid$:
\begin{align}
 \Pesid & \coloneqq\{P_{U,X} \in \cP(\cU \times \cX):U - X - (Y,Z), \, I(U;Y)\geq I(U;Z),\, P_X \circ W_{Z|X}=Q_Z \} \nonumber \\
& = \{P_{U,X} \in \cP(\cU \times \cX):U - X - (Y,Z), \, I(U_1;Y_1) +I(U_2 ; Y_2)\geq I(U_2 ;Z), \nonumber \\
& \, \, \  \, \, \, P_X \circ W_{Z|X}=Q_Z \}.
\label{simul.cond}
\end{align}

Therefore, we have
\begin{align}
  (1-\epsilon)+(H_b(q)-H_b(\delta))(1-\epsilon)&\geq H_b(q)-H_b(\delta)
  \,,
\end{align}
and it follows that
\begin{equation}
  H_b(q)-H_b(\delta)\leq \frac{1}{\epsilon}-1
  \,.
  \label{ConditionReverse}
\end{equation}
If there is no preprocessing and in the absence of additional randomization, i.e., for $\delta=0$ and $U=X$,
the condition reduces to
\begin{equation}
 H_b(q)\leq \frac{1}{\epsilon}-1.
\end{equation}
We note that \cref{ConditionReverse} is satisfied if $\epsilon \le
\frac 1 2$, and also if $q$ satisfies the last inequality.
In this case, one can omit any preprocessing by setting $\delta = 0$.
In the case when $\epsilon > \frac 1 2$ and $H_b(q) > \frac 1 \epsilon - 1$,
preprocessing is needed to satisfy \cref{ConditionReverse}, i.e., $\delta$ must
be strictly larger than $0$.
Hence, without preprocessing, simulating $Q_Z$ is impossible for certain combinations of $\epsilon$ and $q$.

Computing $R_{\text{L}}$, as defined in \cref{BoundsDefn}, we get the lower bound on $\capESID$ by:
\begin{align}
\label{RL.ineq} \capESID &\geq R_L = \max_{P_{U,X}\in\Pesid} I(U;Y)  \\
&\overset{(a)}{=}\max_{P_{X_1}\in\cP(\cX_1)} I(X_1;Y_1)+\max_{P_{U_2,X_2}\in\Pesid} I(U_2;Y_2)  \nonumber \\
&\overset{(b)}{=} (1-\epsilon)+ \max_{P_{U_2,X_2}\in\mathcal{P}_{21}} I(U_2;X_2) (1-\epsilon) \label{optimfirst} \nonumber \\
&\overset{(c)}{=} (1-\epsilon)+\max_{\delta\in \mathcal{D}}\brack{H_b(q)-H_b(\delta))}(1-\epsilon) \nonumber \\
&\overset{(d)}{=}(1-\epsilon)+\brack{H_b(q)-\max\set{ \textstyle H_b(q)+1-\frac 1 \epsilon,\,0 }}(1-\epsilon) \nonumber \\
&\overset{}{=}(1-\epsilon)+\min \set{ H_b(q),\,1 - \textstyle \frac 1 \epsilon }(1-\epsilon),
\end{align}
where \cref{RL.ineq} is given by \cref{corollary:achiev.U}; (a) is due to the
independent and parallel structure of the channel, (b) is justified by viewing
the second subchannel as a concatentation of a preprocessing part and a BEC,
(c) holds by considering the set \cref{simul.cond} and defining
$\cD \coloneqq \{\delta: \delta\in[0, \frac 1 2],\,(1-\epsilon)>[H_b(q)-H_b(\delta)]-(1-\epsilon)[H_b(q)-H_b(\delta)]\}$,
and (d) is due to the condition \cref{ConditionReverse}.

Computing $R_{\text{L}, U=X}$, as defined in \cref{BoundsDefn}, the lower bound on $\capESID$ with no preprocessing:
\begin{align}
  \label{R.LUX}\capESID
  &\overset{(a)} \geq R_{\text{L}, U=X}
                 = \max_{P_{X}\in\Paost} I(X;Y) \nonumber  \\
  &\overset{(b)} = \max_{P_{X_1}\in\cP(\cX_1)} I(X_1;Y_1)+\max_{P_{X_2}\in\Paost} I(X_2;Y_2)  \nonumber \\
  &\overset{(c)} = \big[(1 - \epsilon) + H_b(q) (1 - \epsilon)\big]\ind{H_b(q) \leq \textstyle \frac 1 \epsilon - 1},
\end{align}
where (a) holds by \cref{corollary:achiev.U},
(b) holds due to the independence of the two subchannels, and (c) is true by \cref{ConditionReverse}.
The condition \cref{ConditionReverse} implies that $R_{L,U=X}$ is zero for some
range of $q$ if no preprocessing is used, for channels with $\epsilon > \frac 1 2$.
Hence, we have a case where randomized preprocessing via the auxiliary channel $P_{X|U}$
not only improves the achievable rate, but it also essential to allow any
possibility of any positive-rate ESID code.

Computing $R_U$,  as defined in \cref{BoundsDefn}, we obtain the upper bound
\begin{align}
\label{R.U}\capESID &\leq R_U = \max_{P_{X}\in\Pesid} I(X;Y) \nonumber \\
&=\max_{P_{X_1}\in\cP(\cX_1)} I(X_1;Y_1)+\max_{P_{X_2}\in\mathcal{P}_{\text{ESID},2}} I(X_2;Y_2)  \nonumber \\
& = \tup{ 1 + H_b(q) } \tup{ 1 - \epsilon }.
\end{align}
The ESID capacity $\capESID(W_{Y|X},W_{Z|X})$ is visualized in
\cref{fig:ReverseA,fig:ReverseB,fig:ReverseC} and contrasted to different
identification and transmission capacities, with and without secrecy and stealth
constraints. In \cref{fig:ReverseA} where $\epsilon=0.4$, the lower and upper
bounds on $\capESID$ coincide and thus $\capESID$ is exactly specified. In
\cref{fig:ReverseB,fig:ReverseC}, $\capESID$ is not known as the lower and upper
bounds do not match.
\begin{remark}
  Notice that, in this example,
  the upper and lower bounds to $\capESID$ coincide when $\epsilon \le \frac 1 2$,
  even if the channel is not more capable as in \cref{corollary:esid.dm.moreCapable}.
  This highlights that the condition in \cref{corollary:esid.dm.moreCapable} is sufficient but not necessary,
  and there is actually a larger class of channels for which \cref{corollary:achiev.U,thm:esid.dm}
  together specify the exact capacity.
  Actually, it is sufficient that $I(X; Y) \ge I(X; Z)$ for the input distribution $P_X$
  that maximizes $I(X; Y)$.
\end{remark}

\section{Proofs}

In this section, we give the proofs of the main results of our paper.
\label{sec:esid.dm.proof}

\subsection{Proof of \cref{AOST.Theorem}}

\subsubsection{Direct Part}
First, we prove that the rate is achievable, by randomly constructing a code and
analyzing its error probabilities.

\paragraph{Codebook generation}
Fix a $P_X \in \Paost \subseteq \cP(\cX)$,
\begin{gather}
  \label{SimulationCondition2}
  P_X\circ W_{Z|X} = Q_Z
  \,,
\shortintertext{and for an any $0 < \epsilon \le (I(X; Y) - I(X; Z))/(2 H(Z))$, let}
  I(X; Z) + 2 \epsilon H(Z) \le R_{\text{AOST}} \le I(X; Y)
  .
\end{gather}
Randomly and independently, generate a codebook $C_o = \tup{ X^n(1), \dots, X^n(\card{\cM_{\text{AOST}}}) }$
of $|\mathcal{M}_{\text{AOST}}|=2^{n R_{\text{AOST}}}$ codewords,
where $X^n(m) \sim P_X^n$.

\paragraph{Encoding and decoding}
Alice encodes a message $m \in \cM_{\text{AOST}}$
as $X^n(m)$ and sends this into the channel $W_{Y,Z|X}^n$.
Bob receives $Y^n$ and tries to find a unique $\hat{m}$ such
that
\begin{equation}
  (X^n(\hat{m}), Y^n) \in \cT_\epsilon^n (P_{X,Y})
  .
\end{equation}
If there is such an $\hat{m}$, he declares that $m = \hat{m}$, i.e.,
he sets $\hat M = \hat m$.
Otherwise, he declares that $m = 1$, i.e., he sets $\hat M = 1$.

\paragraph{Error analysis}

\begin{align}
  E_1 &= \Pr( M \ne \hat M | {C_o}) \,, \\
  E_2 &=  D(P_{Z^n|C_o}\|\,Q_Z^n)  \,,
\end{align}
where by the union bound and Markov's inequality,
\begin{align}
  P_{C_o}\tup{ E_1 \ge \epsilon^{(n)} \cup E_2 \ge \delta^{(n)} }
     &\le P_{C_o}\tup{ E_1 \ge \epsilon^{(n)} }
        + P_{C_o}\tup{ E_2 \ge \delta^{(n)} }
  \\ &\le \frac 1 {\epsilon^{(n)}} \expect[E_1]
        + \frac 1 {\delta^{(n)}}   \expect[E_2]
        \label{AchievabilityStealth1}
  .
\end{align}
Using standard typicality arguments~\cite[Chap.~3]{ElGamalKim2011NetworkInformationTheory},
one can show that
\begin{equation}
  \expect[E_1] \to 0
  \text{ as }
  n \to \infty
  ,
\end{equation}
since $R_{\text{AOST}} \le I(X;Y)$,
and hence $\expect[E_1]/\epsilon^{(n)} \to 0$ in \cref{AchievabilityStealth1},
for some $\epsilon^{(n)}$ where $\epsilon^{(n)} \to 0$.

Convergence of $\expect[E_2]$ follows from~\cite{HouKramer2013divergence}:
By~\cite[Eqs. 10]{HouKramer2013divergence},
\begin{align}
  \expect[E_2]
     &= \expect_{P_{C_o}} \brack{ D( P_{Z^n|C_o}\| Q_Z^n ) }
  \\ &\le
    \expect_{P_{X^n, Z^n}} \brack{
        \log\tup{ \frac { W_{Z|X}^n (Z^n|X^n(M)) } { |\mathcal{M}|\, Q_Z^n(Z^n) } + 1 }
      }
  \label{eq:AOSbound}
  .
\end{align}
As the derivation in~\cite{HouKramer2013divergence}
is quite complicated, we present one in \cref{appendix:AOSbound}
that is adapted to our notation.
Since $R_{\text{AOST}} \ge I(X;Z) + 2\epsilon H(Z)$,
it follows from typicality arguments~\cite[Eqs. 11-17]{HouKramer2013divergence}
for some $\delta^{(n)}$ where $\delta^{(n)} \to 0$,
\begin{equation}
  \frac {\expect[E_2]} {\delta^{(n)}}
  = \expect_{P_{C_o}} \brack{D( P_{Z^n|C_o}\| Q_Z^n )} / \delta^{(n)}
  \to 0
  .
\end{equation}
Since both expected errors converge, it follows from \cref{AchievabilityStealth1}
that there exists for every $n$ a realization $c_o$ of $C_o$ that satisfies
\begin{align}
  \Pr( M \ne \hat M | C_o = c_o) &\le \epsilon^{(n)} \to 0, \\
  D(P_{Z^n|c_o} \| Q_Z^n) &\le \delta^{(n)} \to 0,
\end{align}
i.e., our coding scheme is for $C_o = c_o$ a $(\card{\cM_{\text{AOST}}}, n |
\epsilon^{(n)}, \delta^{(n)})$ code for $(W_{Y|X}, W_{Z|X})$ and reference
distribution $Q_Z^n$,
where both error bounds converge. Hence
all rates
\begin{equation}
I(X; Z) + 2 \epsilon H(Z) \le R_{\text{AOST}} \le I(X; Y)
\label{thm1.achievability}
\end{equation}
are achievable, for any $P_X \in \Paost$.

\subsubsection{Converse Part}

\begingroup
\allowdisplaybreaks

To prove that the lower bound from the previous part is
also an upper bound, we start with a standard transmission converse based
on Fano's inequality and single-letterization.
We assume that $(f_{\mathsf{AOST}}, \phi)$
is a $(2^{n R_{\text{AOST}}}, n | \lambda^{(n)}, \delta^{(n)})$ code,
where $\lambda^{(n)}, \delta^{(n)} > 0$
and $\lim_{n \to \infty} \max\set{\lambda^{(n)}, \delta^{(n)}} = 0$.
Thus, when $M$ is the randomly selected message,
$X^n = f_{\text{AOST}}(M)$, and there exists $\zeta_n > 0$
such that $\zeta_n \to 0$ as $n \to \infty$ and
\begin{align}
  n R_{\text{AOST}}
     &= H(M) = I(M;Y^n)+ H(M|Y^n)
  \\ &\overset{(a)} \le I(M; Y^n) + n\zeta_n
  \nonumber
  \\ &\overset{(b)} \le I(X^n; Y^n) + n\zeta_n
  \\ &= H(Y^n) - H(Y^n|X^n) + n\zeta_n
  \nonumber
  \\ &\overset{(c)} = \sum_{i=1}^n H(Y_i|Y^{i-1}) - \sum_{i=1}^n H(Y_i|X^n, Y^{i-1}) + n\zeta_n
  \nonumber
  \\ &\overset{(d)} \le \sum_{i=1}^n [H(Y_i) - H(Y_i|X_i)] + n\zeta_n
  \nonumber
  \\ &= n \sum_{i=1}^n \frac 1 n I(X_i; Y_i) + n\zeta_n
  \nonumber
  \\ &\overset{(e)} \le n (I(X; Y) + \zeta_n)
  \label{ConverseStealth1}
  ,
\shortintertext{where}
  P_X(x)
     &= \frac 1 n \sum_{i=1}^n P_{X_i}(x)
  \\
  P_{Y}(y)
     &= (P_X \circ W_{Y|X}) (y)
      = \frac 1 n \sum_{i=1}^n \sum_x P_{X_i}(x) W_{Y|X}(y | x)
  \label{single.letterization}
\end{align}
where (a) follows from Fano's inequality;
(b) follows from the data processing inequality;
(c) is the chain rule for the (conditional) entropy;
(d) holds since the function $y^{i-1} \mapsto H(Y_i|Y^{i-1} = y^{i-1})$ is concave in $P_{Y^{i-1}}$,
    and since $(X^{ i - 1}, X_{i+1}^{n}, Y^{i-1} ) - X_i- Y_i $ is a Markov chain;
and (e) follows from the concavity of the mutual information in the channel input distribution.

Now, we have shown that
\begin{equation}
R_{\text{AOST}} \le I(X; Y) + \zeta_n \to I(X; Y).
\end{equation}
What remains to be proven is that
$P_X \circ W_{Y|X} = Q_Z$ and $I(X; Y) \ge I(X;Z)$ are necessarily satisfied.
To this end, we first show that $R_{\text{AOST}} \ge I(X; Z) + \delta^{(n)}/n$
using arguments similar to \cite[Section 5.2.3]{Hou14thesis}.
Since an AOST-sender selects $M$ with uniform distribution,
and $X^n$ is a function of $M$,
\begin{align}
  n R_{\text{AOST}}
     &= H(M)
  \nonumber
  \\ &\ge H(X^n)
\ge I(X^n;Z^n)
  \nonumber
  \\ &\overset{(a)} \ge I(X^n;Z^n)+D(P_{Z^n}\|Q_Z^n)-\delta^{(n)}
  \nonumber
  \\ & = H(Z^n)-H(Z^n|X^n) - \expect_{P_{Z^n}}\big(\log{Q^{n}_{Z}(Z^n)}\big) -  H(Z^n) -\delta^{(n)}
  \nonumber
  \\ &\overset{(b)}  = - \sum_{i=1}^n H(Z_i| Z^{i-1}, X^n) -  \sum_{i=1}^n \expect_{P_{Z_i}}\big(\log{Q_{Z}(Z_i)}\big) -\delta^{(n)}
  \nonumber
  \\ & \overset{(c)} = - \sum_{i=1}^n H(Z_i| X_i) - \sum_{i=1}^n \expect_{P_{Z_i}}\big(\log{Q_{Z}(Z_i)}\big) -\delta^{(n)}
  \nonumber
  \\ & = \sum_{i=1}^n \sum_{x, z \::\: Q_Z(z) > 0} P_{X_i}(x)W_{Z|X}(z |x) \bigg(\log\frac{W_{Z|X}(z|x)}{Q_Z(z)}\bigg) -\delta^{(n)}
  \nonumber
  \\ & \ge n \sum_{x, z \::\: P_Z(z) \cdot Q_Z(z) > 0} P_X(x) W_{Z|X}(z |x)
       \bigg(\log\frac{W_{Z|X}(z|x)}{P_Z(z)}\bigg)
     \nonumber\\
     &\quad + n \sum_{x, z \::\: Q_Z(z) > 0} P_X(x) W_{Z|X}(z |x) \log\frac {P_Z(z)} {Q_Z(z)}
            - \delta^{(n)}
  \nonumber
  \\ & = n D( P_{X, Z} \| P_X P_Z ) + n D(P_Z \| Q_Z) - \delta^{(n)}
  \nonumber
  \\ &\overset{(d)} \ge n I(X ; Z) - \delta^{(n)}, \label{ConverseStealth2}
\end{align}
where
$P_Z = P_X \circ W_{Z|X}$,
(a) holds as $D(P_{Z^n} \| Q_Z^n) \le \delta^{(n)}$
    by the definition of an AOST code (\cref{AOSTcodeDef});
(b) follows from the chain rule and since $Q_Z^n$ is i.i.d;
(c) holds since $(X^{ i - 1}, X_{i+1}^n, Z^{i-1} ) - X_i - Z_i $ is a Markov chain;
and (d) follows from the definition of the mutual information, and because the KL-divergence is positive.
Thus, by \cref{ConverseStealth1,ConverseStealth2} and considering the limit for $n \to \infty$,
we get that if $R_{\text{AOST}}$ is an achievable rate, it satisfies that
\begin{equation}
  \label{eq:ConverseStealth.mutinfOrder}
  I(X; Y) \ge R_{\text{AOST}} \ge I(X; Z)
.
\end{equation}

Finally, we show that $P_Z = Q_Z$ must be satisfied.
Since $\lim_{n\to\infty} \delta^{(n)} = 0$,
this follows from the single-letterization in \cite[Eq. 5.29]{Hou14thesis},
which holds actually for all $P_{Z^n} \in \cP(\cZ^n)$, namely,
\begin{align}
  \delta^{(n)} &\geq D(P_{Z^n} \| Q_Z^{n}) \nonumber \\
  & \overset{(a)}{=}  -\expect_{P_{Z^n}}\log \big({Q_Z^n}\big) - H(Z^n) \nonumber \\
  & \overset{(b)}{\geq} -\sum_{i=1}^n \expect_{P_{Z_i}} \big(\log{Q_Z(Z)}\big) - \sum_{i=1}^n H(Z_i) \nonumber \\
  & \overset{(c)}{=}  n \sum_{i=1}^n \frac 1 n D(P_{Z_i} \| Q_Z) \nonumber\\
  & \overset{(d)}{\geq} n \, D(P_{Z} \| Q_Z),
  \label{ConverseStealth3}
\end{align}
where (a) and (c) hold by the definition of KL-divergence;
(b) stems from the fact that $Q_Z^n$ is i.i.d., the chain rule of the entropy,
    and the fact that conditioning reduces the entropy;
and (d) follows from the joint convexity of the KL-divergence.
Therefore, any achievable rate satisfies that
$D (P_Z\|Q_Z) \le \lim_{n \to \infty} \delta^{(n)}/n = 0$, and together
with \cref{eq:ConverseStealth.mutinfOrder}, this implies that
every achievable rate $R_{\text{AOST}}$ satisfies
\begin{equation}
  R_{\text{AOST}} \le I(X; Y),
\end{equation}
where $I(X; Y) \ge I(X; Z)$ and $P_X \circ W_{Z|X} \eqqcolon P_Z = Q_Z$.
Together with the direct part as concluded in \cref{thm1.achievability}, this proves \cref{AOST.Theorem}.
\qed

\endgroup

\subsection{Proof of \cref{prop:achiev.X} }

The main idea is very similar to the achievability proofs in
\cite{AhlswedeDueck1989Identificationpresencefeedback,AhlswedeZhang1995Newdirectionstheory},
where an identification code with message set $\cM$ is composed of two transmission codes,
one of them a $(|\cM_1|, n \,| \, \lambda^{(n)}, \delta^{(n)})$ code and
the other a $(|\cM_2|, k \, |\, \lambda^{(k)}, \delta^{(k)})$ code,
and additionally a $\epsilon$-almost universal hash function,
i.e., a function $t : \cM \times \cM_1 \to \cM_2$ such that
for $M_1 \sim P^{\mathbf{unif}}_{\cM_1}$ and every distinct pair $m,m' \in \cM$, $m \ne m'$,
\begin{equation}
  \Pr\set{
    t_m(M_1) = t_{m'}(M_1)
  }
  \le \epsilon
  .
\end{equation}
Hence, for $k$ much smaller than $n$, e.g., $k = \ceil{\sqrt{n}\,}$,
Alice samples the seed $M_1$ and encodes the identification message $m$
into the pair $(M_1, M_2)$ by hashing $M_2 = t_m(M_1)$.
This pair is already an identification codeword for a noiseless channel.
Then, she transmits $M_1$ with the first transmission code and $M_2$ with the second one.
Bob, after selecting a message $m'$ and decoding the pair $(M_1, M_2)$ of
transmission messages, hashes his message with the seed $M_1$, i.e., he computes $M'_2 = t_{m'}(M_1)$.
Then he tests if $M_2 = M'_2$ and declares that $m = m'$ if and only if $M_2 = M_2'$.
The probability of second-type error of this code is
the probability of a hash collision, i.e., of $M_2 = M'_2$ in case $m \ne m'$.
This is upper-bounded by $\epsilon$, according to $\epsilon$-almost universality.

The key to applying this in our ESID setting is to use two transmission codes
with different security properties: in \cite{AhlswedeZhang1995Newdirectionstheory},
$M_1$ is transmitted with a normal transmission code,
and $M_2$ is transmitted with a secret transmission code. Since $M_1$ carries no
information about $m$, the latter is protected against an eavesdropper.
Here, in ESID, we additionally require both codes to be stealthy,
and thus $M_1$ is transmitted with an AOST code, and $M_2$ with an
effectively secret transmission code, thus maintaining stealth on the whole
codeword, and secrecy only for the hash $M_2$.

In the following, we formally construct such an identification code and prove
that it is an ESID code
as long as the rate satisfies the bound in \cref{prop:achiev.X}.

\subsubsection{ESID Code Construction}
In order to streamline the proof, we formalize a lemma implicitly introduced in~\cite[Sec.~III]{AD89feedback} and used in  various achievability proofs of identification problems (e.g.~\cite[Prop.~16]{BocheDeppe2018SecureIdentificationWiretap}, \cite[Thm.~5.6]{BDW19}, \cite[Thm.~17]{LBDW21}).

\begin{lemma}
  \label{lemma:uhf.exists}
  There exists a $\epsilon$-almost-universal hash function $t : \cM \times \cM_1 \to \cM_2$ as defined above,
  provided that $\epsilon > \frac 1 {\cM_2}$ and
  \begin{equation}
    \label{eq:uhf.exists.condition}
    \card\cM^2 < 2^{\card{\cM_1} (\epsilon \log \card{\cM_2} - H_b(\epsilon))}
    .
  \end{equation}
\end{lemma}
\begin{proof}
  We randomly construct a random version, $T$, of $t$ by sampling every
  $T_m(m_1)$, $m \in \cM$, $m_1 \in \cM_1$,
  from the uniform distribution over $\cM_2$.
  Thus, for every distinct pair $m, m' \in \cM$ (where $m \ne m'$) and every $m_1 \in \cM_1$,
  \begin{align}
    \expect \ind{ T_m(m_1) = T_{m'}(m_1) }
       &= \sum_{m_2, m_2' \in \cM_2} \frac 1 {\card{\cM_2}^2} \ind{ m_2 = m_2'}
    \\ &= \sum_{m_2 \in \cM_2} \frac 1 {\card{\cM_2}^2}
        = \frac 1 {\card{\cM_2}}
    \,.
  \end{align}
  Therefore, by a Chernoff bound (\cref{chern}),
  \begin{align}
    &\Pr \brack{
      \exists\, m, m' \in \cM, m \ne m' : P_{M_1} \brack{ T_m(M_1) = T_{m'}(M_1) } \ge \epsilon
    }
    \nonumber
    \\ &= \sum_m \sum_{m' \ne m} \Pr \brack[\Big]{
            \sum_{m_1} \ind{ T_m(m_1) = T_{m'}(m_1) } \ge \card{\cM_1} \epsilon
          }
    \\ &\le \card \cM^2 \cdot 2^{ -\card{\cM_1} D_b\tup{ \epsilon }[\| \frac 1 {\card{\cM_2}}] }
    \\ &\le \card \cM^2 \cdot 2^{ -\card{\cM_1} (\epsilon \log_2 \card{\cM_2} - H_b(\epsilon)) }
    \\ &< 1
    ,
  \end{align}
  where the last inequality holds by~\cref{eq:uhf.exists.condition},
  and the second last one since $\log \tup{ 1 - \textstyle \frac 1 a } \le 0$ and
  \begin{align}
    D_b(\epsilon \, \| \, a^{-1})
       &= \epsilon \log (\epsilon a) + (1-\epsilon) \log \frac {1 - \epsilon} {1 - \frac 1 a}
    \\ &= \epsilon \log a
          + \epsilon \log \epsilon
          + (1-\epsilon) \log (1 - \epsilon) -  (1-\epsilon) \log \tup{1 - \textstyle \frac 1 a}
    \\ &= \epsilon \log a - H_b(\epsilon)
          -  (1-\epsilon) \log \tup{ 1 - \textstyle \frac 1 a }
    \\ &\ge \epsilon \log a - H_b(\epsilon)
    \,.
  \end{align}
  Hence, there exists a realization $t$ of $T$ such that
  \begin{equation}
    P_{M_1} \brack{
      t_m(M_1) = t_{m'}(M_1)
    }
    \le \epsilon
    ,
  \end{equation}
  i.e., $t$ is a $\epsilon$-universal hash function.
\end{proof}

Now, suppose that both $\capEST(W_{Y|X},W_{Z|X},Q_Z)$ and $\capAOST(W_{Y|X},W_{Z|X},Q_Z)$ are positive.
Then, for all
\begin{align}
  R &< \capAOST(W_{Y|X},W_{Z|X},Q_Z),
  \label{eq:esid.direct.Raost}
  \\
  R_{\mathsf{EST}} &< \capEST(W_{Y|X},W_{Z|X},Q_Z),
\end{align}
there exist a $(2^{nR}, n \,| \, \lambda^{(n)}, \delta^{(n)})$-AOST code
with encoder $f_{\text{AOST}}$ and decoder $\phi_{\text{AOST}}$,
and a $(2^{k R_{\mathsf{EST}}}, k \, |\, \lambda^{(k)}, \delta^{(k)})$-EST code
with encoding distributions $F^{\mathsf{EST}}_{m_2}$, $m_2 \in \cM_2$, and decoder $\psi_{\mathsf{EST}})$,
where $k = \ceil{\sqrt {n}}$.
Further, by \cref{lemma:uhf.exists}, there exists
a $\epsilon$-almost universal hash function $t : \cM \times [2^{nR}] \to [2^{k R_{\mathsf{EST}}}]$
which satisfies that $\epsilon = 2^{-kR_{\mathsf{EST}} + 1}$
and that
\begin{align}
  \card\cM^2
     &\ge 2^{2^{nR} (2^{-k R_{\mathsf{EST}} + 1} k R_{\mathsf{EST}} - 2)}
  \\ &\ge 2^{2^{nR} \cdot 2 (k R_{\mathsf{EST}} 2^{-k R_{\mathsf{EST}}} - 1)}
  \,,
\shortintertext{and hence that}
  \card\cM
     &\ge 2^{2^{nR} (k R_{\mathsf{EST}} 2^{-k R_{\mathsf{EST}}} k R_{\mathsf{EST}} - 1)}
     \to 2^{2^{nR}}
  \label{eq:esid.direct.Mlb}
\end{align}
as $n \to \infty$.

For the ESID code,
consider the decoding sets
$\cG_{m_1} = \set{ y^n : \phi_{\text{AOST}}(y^n) = m_1 }$
and $\cD_{m_2} = \set{ y^k : \psi_{\mathsf{EST}}(y^k) = m_2 }$,
and let the encoder and decoder be given by
\begin{align}
  E_m(x^n, x_{n+1}^{n+k})
     &= \sum_{m_1} \frac 1 {\card{\cM_1}} \ind{x^n = f_{\text{AOST}}(m_1) } F^{\mathsf{EST}}_{t_m(m_1)}(x_{n+1}^{n+k})
     \,,
  \\
  \cI_{m'}
     &= \set{ (y^n, y_{n+1}^{n+k}) :
       t_{m'}(\phi_{\text{AOST}}(y^n)) = \psi_{\mathsf{EST}}(y_{n+1}^{n+k})
     }
  \\ &= \bigcup_{m_1 \in \cM_1} \cG_{m_1} \times \cD_{t_{m'}(m_1)}
  \,.
\end{align}

\subsubsection{Missed-Identification Error}

For all $m \in \cM_{\text{ESID}}$, the missed-identification error is bounded by
\begin{align}
     &E_m \circ W_{Y|X}^{n+k}(\cI_m^c)
  \\ &= \frac 1 {|\cM_1|} \sum_{m_1 \in \cM_1} \sum_{x^{n+k}_{n+1} \in \cX^k}
        F_{t_m(m_1)}^{\mathsf{EST}}(x_{n+1}^{n+k})
        W_{Y|X}^{n+k}(\cI_m^c|f_{\text{AOST}}(m_1),x_{n+1}^{n+k})
  \\ &\le \frac 1 {|\cM_1|} \sum_{m_1 \in \cM_1}
        \brack{
          W_{Y|X}^n(\cG_{m_1}^c | f_{\text{AOST}}(m_1))
          + \tup{ F^{\mathsf{EST}}_{t_m(m_1)} \circ W_{Y|X}^k }
            \tup{ \cD_{t_{m}(m_1)}^c }
        }
  \\ &\le \lambda^{(n)} + \lambda^{(k)} \eqqcolon \lambda_1^{(n+k)},
  \label{eq:esid.direct.typeI}
\end{align}
where the inequality comes from the error bounds of both the AOST and EST codes, respectively.

\subsubsection{False-Identification Error}

In order to bound the false-identification error,
notice that for all $m, \tilde m \in \cM$ with $m \ne \tilde m$,
since $\cG_{m_1} \cap \cG_{m_1'} = \emptyset$ whenever $m_1 \ne m_1'$,
and $\cD_{m_2} \cap \cD_{m_2'} = \emptyset$ whenever $m_2 \ne m_2'$,
the false-identification error is bounded by
\begin{align}
  &E_m \circ W_{Y|X}^{n+k} (\cI_{m'})
  \nonumber\\
    &= \frac 1 {\card{\cM_1}} \sum_{\substack{m_1 \in \cM_1, \\ x_{n+1}^{n+k} \in \cX^k}}
        F^{\mathsf{EST}}_{t_m(m_1)}(x_{n+1}^{n+k})
        W_{Y|X}^{n+k} \tup{
          \bigcup_{m_1' \in \cM_1} \cG_{m_1'} \times \cD_{t_{m'}(m'_1)}
        }[| f_{\text{AOST}}(m_1), x_{n+1}^{n+k}]
  \\ &\le \frac 1 {\card{\cM_1}} \sum_{\substack{m_1 \in \cM_1, \\ x^{n+k}_{n+1} \in \cX^k}}
        F^{\mathsf{EST}}_{t_m(m_1)}(x_{n+1}^{n+k})
        W_{Y|X}^{n+k} \tup{
          \cG_{m_1} \times \cD_{t_{m'}(m_1)}
        }[| f_{\text{AOST}}(m_1), x_{n+1}^{n+k}]
  \nonumber\\ &\quad + \frac 1 {\card{\cM_1}} \sum_{\substack{m_1 \in \cM_1, \\ x_{n+1}^{n+k} \in \cX^k}}
        F^{\mathsf{EST}}_{t_m(m_1)}(x^{n+k}_{n+1})
        W_{Y|X}^{n+k} \tup{
          \cG_{m_1}^c \times \cY^k
        }[| f_{\text{AOST}}(m_1), x_{n+1}^{n+k}]
  \\ &\le \frac 1 {\card{\cM_1}} \sum_{m_1 \in \cM_1}
        F^{\mathsf{EST}}_{t_m(m_1)} \circ
        W_{Y|X}^k \tup{ \cD_{t_{m'}(m_1)} }
        + \lambda^{(n)}
  \\ &= \lambda^{(n)}
      + \frac 1 {\card{\cM_1}} \sum_{\substack{m_1 \in \cM_1 : \\ t_m(m_1) \ne t_{m'}(m_1)}}
        F^{\mathsf{EST}}_{t_m(m_1)} \circ W_{Y|X}^k \tup{ \cD_{t_{m'}(m_1)} }
  \nonumber\\ &\hspace{3.2em} + \frac 1 {\card{\cM_1}} \sum_{\substack{m_1 \in \cM_1 : \\ t_m(m_1) = t_{m'}(m_1)}}
        F^{\mathsf{EST}}_{t_m(m_1)} \circ W_{Y|X}^k \tup{ \cD_{t_{m'}(m_1)} }
  \\ &\le \lambda^{(n)}
      + \frac 1 {\card{\cM_1}} \sum_{\substack{m_1 \in \cM_1 : \\ t_m(m_1) \ne t_{m'}(m_1)}}
        F^{\mathsf{EST}}_{t_m(m_1)} \circ W_{Y|X}^k \tup{ \cD_{t_m(m_1)}^c }
  \nonumber\\ &\hspace{3.2em} + \frac 1 {\card{\cM_1}} \sum_{m_1 \in \cM_1}
      \ind{t_m(m_1) = t_{m'}(m_1)}
  \\ &\le \lambda^{(n)} + \lambda^{(k)} + \epsilon
       = \lambda^{(n+k)}_1 + \epsilon
       \eqqcolon \lambda_2^{(n+k)}
  \label{eq:esid.direct.typeII}
  \,,
\end{align}
where the last inequality holds by the definition of an EST code with
error probability $\lambda^{(k)}$, and from the definition of
an $\epsilon$-almost universal hash function.

\subsubsection{Leakage}

To bound the leakage, we define
\begin{align}
  P_{Z^{n+k}|m} &= E_m \circ W_{Z|X}^{n+k},
  \\
  P_{Z^n|m}(z^n) &= P_{Z^n}(z^n) = \sum_{m_1} P^{\mathbf{unif}}_{\cM_1} (m_1)
       W^n_{Z|X}(\cdot|f_{\text{AOST}}(m_1)),
  \\
  P_{Z^{n+k}_{n+1}|M_1, m}(\cdot|m_1) &= F^{\mathsf{EST}}_{t_m(m_1)} \circ W_{Z|X}^k
  \,.
\end{align}
Thus, we have
\begin{align}
     & D(E_m \circ W_{Z|X}^{n+k} \| Q_Z^{n+k})
      = D(P_{Z^{n+k}|m} \| Q_Z^{n+k})
  \nonumber
  \\ &\overset{(a)} = D(P_{Z^n|m} \| Q_Z^n) + D(P_{Z^{n+k}_{n+1}|Z^n,m} \| Q_Z^k \,|\, P_{Z^n|m})
  \nonumber
  \\ &\overset{(b)} =
      D (P_{Z^n} \| Q_Z^n)
      + D (P_{M_1|Z^n} P_{Z^{n+k}_{n+1}|M_1, m} \| Q_Z^k \,|\, P_{Z^n})
  \\ &\le D (P_{Z^n} \| Q_Z^n)
      + D (P_{Z^{n+k}_{n+1}|M_1, m} \| Q_Z^k \,|\, P_{M_1})
  \\ &\le D (P_{Z^n} \| Q_Z^n)
      + \sum_{m_1 \in \cM_1} \frac 1 {\card{\cM_1}} D (F^{\mathsf{EST}}_{t_m(m_1)} \circ W_{Z|X}^k \| Q_Z^k)
  \\ &\le D (P_{Z^n} \| Q_Z^n)
      + \max_{m_1 \in \cM_1} D (F^{\mathsf{EST}}_{t_m(m_1)} \circ W_{Z|X}^k \| Q_Z^k)
  \\ &\le \delta^{(n)} + \delta^{(k)} \coloneqq \delta^{(n+k)}
  \label{eq:esid.direct.leakage}
  ,
\end{align}
where (a) follows from the chain rule of KL-divergence \cref{KLchainrule},
while (b) holds because $Q_Z^n$ is an i.i.d. distribution
and since $Z^{n+k} = (Z^n, Z^{n+k}_{n+1})$ and $Z^n - M_1 - Z^{n+k}_{n+1}$.
Thus, since $k = \sqrt n$, by
\cref{eq:esid.direct.Raost,eq:esid.direct.Mlb,eq:esid.direct.typeI,eq:esid.direct.typeII,eq:esid.direct.leakage},
for every $R < \capAOST(W_{Y|X},W_{Z|X},Q_Z)$
there exists a
$(2^{2^{n R}}, n \,|\, \lambda^{(n+\sqrt n)}_1, \lambda^{(n+\sqrt n)}_2, \delta^{(n+\sqrt n)})$-ESID code
$\set{ (E_m, I_m) : m \in \cM }$,
where
\begin{equation}
  \lim_{n\to \infty} \max{\lambda^{(n+\sqrt n)}_1, \lambda^{(n+\sqrt n)}_2, \delta^{(n)})} = 0
  .
\end{equation}
This concludes the proof of \cref{prop:achiev.X}.
\qed

\subsection{A Sharpened Identification Converse}
\label{sec:auxResults}

To prove an upper rate bound for ESID in \cref{sec:esid.dm.proof},
we will use the following converse proof for ID codes.
It is a sharpened version of Watanabe's bound (see
Corollary~2 and Lemma~1 in \cite{Watanabe2022MinimaxConverseIdentification}).
The difference is that in Watanabe's bound, the
maximization is over the whole set $\cP(\cX^n)$
of all possible input distributions to the channel,
while ours is optimized only over the convex hull $\cE$
of all encoding distributions of the code.
Since our ID codes use stochastic encoding,
there are deterministic encoding distributions, which assign weight one to
particular code words, that are no encoding distributions. Therefore,
$\cE$ is a proper subset of $\cP(\cX^n)$.
Optimizing over $\cP(\cX^n)$ would drop the effective-secrecy constraint from
the upper rate bound, whereas $\cE$ includes the constraint,
because it is satisfied by every encoding distribution.
This is crucial in the proof of \cref{thm:esid.dm} in \cref{sec:esid.dm.proof}.

As the bounds in \cite{Watanabe2022MinimaxConverseIdentification},
our bounds are non-asymptotic and do not depend on the block length $n$.
Therefore, we drop all $n$'s and consider $(|\cM|, \lambda_1, \lambda_2)$-ID codes
for channels with arbitrary memory, as defined in
\cite[Section~II]{Watanabe2022MinimaxConverseIdentification}.
Thus, when applying \cref{lemma:idConverse.oneShot} in \cref{eq:esid.converse.oneShot.rate}
to memoryless channels and sets of code words with block length $n$,
the one-shot set $\cX$ of possible channel inputs
is substituted by the set $\cX^n$ of possible channel input blocks of length $n$.

Consider the hypothesis testing divergence
\begin{gather}
  D_\alpha(P_Y \| Q_Y)
  := \sup\set{ \gamma : \text{Pr}\tup{ \log \frac{P_Y(Y)}{Q_Y(Y)} \le \gamma } \le \alpha },
\end{gather}
where $Y \sim P_Y$, and denote the joint distributions
$P_{X,Y}(x,y) \coloneqq P_X(x) W_{Y|X}(y|x)$ and $P_X Q_Y(x,y) \coloneqq P_X(x) Q_Y(y)$.

\begin{lemma}
  \label{lemma:idConverse.oneShot}
  Let $\lambda_1, \lambda_2, \eta > 0$ and $\alpha := \lambda_1 + \lambda_2 + 2\eta < 1$.
  The size of every $(|\mathcal{M}|, \lambda_1, \lambda_2)$-ID code $\set{(E_m, \cD_m)}_{m \in \mathcal{M}}$
  for a channel $W_{Y|X} : \cX \to \cP(\cY)$
  is bounded by
  \begin{align}
    \log\log |\mathcal{M}|
    &
    \label{eq:idConverse.oneShot.divergence.PX}
    \le \max_{P_X \in \cE} \min_{Q_Y \in \cP(\cY)}
      D_\alpha (P_{X,Y} \| P_X  Q)
      + \epsilon
    \\&
    \label{eq:idConverse.oneShot.mutinf}
    \le \max_{P_X \in \cE}
      \frac{1}{1-\alpha}
      I(X; Y)
      + \epsilon,
  \end{align}
  where $\epsilon = \log\log \card\cX + 3 \log(1/\eta) + 2$,
  and $\cE = \set{ P_M \circ P_{X|M} : P_M \in \cP(\mathcal{M})}$
  is the convex hull of all encoding distributions $P_{X|M=m} = E_m$.
\end{lemma}
\begin{proof}
  By   \cite[Corollary~2 and~Lemma~1]{Watanabe2022MinimaxConverseIdentification},
  we have that
  \begin{gather}
    \label{eq:idConverse.oneShot.divergence.PX.orig}
    \log\log |\mathcal{M}|
    \le \max_{P_X \in \cP(\cX)} \min_{Q_Y \in \cP(\cY)}
      D_\alpha (P_{XY} \| P_X Q_Y)
      + \epsilon.
  \end{gather}
  The maximization over $\cP(\cX)$ is introduced in
  the proof of \cite[Theorem~1]{Watanabe2022MinimaxConverseIdentification}
  to establish the upper bound
  \begin{align}
    \label{eq:watanabe.bound}
    \frac{1}{2} &\brack{ (E_m \times W_{Y|X}) (\cS) + (E_{m'} \times W_{Y|X}) (\cS) }
    \nonumber
    \\
    &\le
    \sup_{P_X \in \cP(\cX)} P_X \times W_{Y|X} (\cS),
  \end{align}
  for some set $\cS \subseteq \cX \times \cY$,
  which ultimately leads to the maximization in
  \cref{eq:idConverse.oneShot.divergence.PX.orig}.
  Clearly, a maximization over all $E_m$, $m \in \mathcal{M}$
  would suffice in \cref{eq:watanabe.bound}.
    To
  establish the minimax equality in \cite[Corollary~2]{Watanabe2022MinimaxConverseIdentification},
  Watanabe used the fact that the supremum is taken over
  a compact and convex set. Since no   other
  properties
  of $\cP(\cX)$ are used in~\cite{Watanabe2022MinimaxConverseIdentification},
  it suffices to maximize over the convex hull $\cE$
  of all encoding distributions,     as in~\cref{eq:idConverse.oneShot.divergence.PX}.
  By Markov's inequality,
  \begin{align}
    D_\alpha(P_Y \| Q_Y) \nonumber
    &
      = \sup\set{ \gamma : \text{Pr}\tup{ \log \frac{P_Y(Y)}{Q_Y(Y)} > \gamma } \ge 1 - \alpha } \nonumber
    \\&
      \le \inf\set{ \gamma : \text{Pr}\tup{ \log \frac{P_Y(Y)}{Q_Y(Y)} \ge \gamma } \le 1 - \alpha } \nonumber
    \\&
      \le \frac{1}{1-\alpha} \expect_{P_Y}\intv{ \log\frac{P_Y(Y)}{Q_Y(Y)} } \nonumber
    \\&
      = \frac{1}{1-\alpha} D(P_Y\|Q_Y),
  \end{align}
  where the first inequality holds since $p(\gamma) = \text{Pr}\tup{\log\frac{P(Y)}{Q(Y)} \ge \gamma}$ is a decreasing function.
  Thus, for every $P_X \in \cP(\cX)$,
  \begin{align}
    \min_{Q \in \cP(\cY)}
    \nonumber
      & D_\alpha(P_X  W_{Y|X} \| P_X\,  Q_Y)
    \\& \overset{(a)}{\le} \frac{1}{1-\alpha} D(P_X  W_{Y|X} \| P_X  P_Y)
    \\& =   \frac{1}{1-\alpha} I(X; Y),
  \end{align}
  where (a) follows as $P_Y= P_X \circ W_{Y|X}$. Then,  \cref{lemma:idConverse.oneShot} follows.
\end{proof}

\subsection{Proof of \cref{thm:esid.dm}}

Consider an arbitrary
$(|\mathcal{M}_{\text{ESID}}|, n \, | \, \lambda_1^{(n)},\lambda_2^{(n)},\delta^{(n)})$-ESID code
$\set{ \tup{E_m, \cD_m} }_{m\in\mathcal{M}_{\text{ESID}}}$
for a DM-WTC $(W_{Y|X},W_{Z|X})$ and a reference output distribution $Q_Z^n$,
where $\lambda_1^{(n)}, \lambda_2^{(n)}, \eta^{(n)} > 0$
satisfy $\alpha^{(n)} \coloneqq \lambda_1^{(n)} + \lambda_2^{(n)} + 2\eta^{(n)} < 1$,
and assume that $\lim_{n\to\infty} \max_{k=1,2} \lambda_k^{(n)} = 0$.
By \cref{lemma:idConverse.oneShot}, the rate is upper-bounded by
\begin{gather}
  \label{eq:esid.converse.oneShot.rate}
  R =   \frac{1}{n} \log\log \card{\cM_{\text{ESID}}}
    \le \max_{P_{X^n} \in \cE}
        \frac{1}{n(1-\alpha^{(n)})} I(X^n; Y^n) + \frac {\epsilon^{(n)}} n,
\end{gather}
where $\epsilon^{(n)} = \log\log \card{\cX^n} + 3\log(1/\eta^{(n)}) + 2$,
and $\cE$ is the convex hull of all encoding distributions.
Similarly to \cref{ConverseStealth1}, we have the single-letterization
\begin{equation}
  \frac 1 n I(X^n; Y^n)
    \le I(X; Y),
  \label{eq:singleLetter.mutInfXY}
\end{equation}

In the following, we single-letterize the constraints on $\cE$.
By \cref{ConverseStealth3},
it holds for every $P_{X^n} \in \cE$ and $P_{Z^n} = P_{X^n} \circ W_{Z|X}^n$
that
\begin{equation}
  D(P_Z \| Q_Z)
  \le \frac 1 n D(P_{Z^n} \| Q_Z^n)
  \le \frac {\delta^{(n)}} n
  ,
  \label{eq:esid.converse.singleLetter.divergence}
\end{equation}
where $P_X = \sum_{i=1}^n P_{X_i}$ and $P_Z = P_X \circ W_{Z|X}$.
Note that since $P_{X^n}$ is in the convex hull $\cE$ of all encoding
distributions $P_{X^n|M}(x^n|m) = E_m(x^n)$, $m \in \cM_{\text{ESID}}$,
we can decompose it as $P_X^n = P_M \circ P_{X^n|M}$ for and some $P_M$.
Thus, similarly to \cite[Eq. (1.49)]{HouKramerBloch2017EffectiveSecrecyReliability},
for sufficiently large $n$,
i.e., sufficiently small
$\lambda_{1}^{(n)}, \lambda_{2}^{(n)}, \delta^{(n)} > 0$,
it holds that
\begin{align}
  \max_{P_M \in \cP(\mathcal{M})} I(M; Y^n)
    & \overset{(a)} \ge \delta^{(n)}
  \nonumber \\
    & \overset{(b)} \ge \max_{P_M \in \cP(\mathcal{M})} D(P_M \circ P_{X^n|M} \circ W_{Z|X}^n \| P_M Q_Z^n )
  \nonumber \\
    & \overset{(c)} \ge \max_{P_M \in \cP(\mathcal{M})} I(M ; Z^n),
\end{align}
where (a) holds by~\cite[Lemma~4]{AhlswedeZhang1995Newdirectionstheory},
(b) is from the definition of the semantic effective secrecy constraint,
and (c) follows similarly to the last two inequalities in \cref{ConverseStealth2}.
Thus, there exists a $P_M \in \cP(\mathcal{M})$ such that
\begin{align}
  0 & \le \frac{1}{n} [ I(M; Y^n) - I(M; Z^n) ] \nonumber \\
& \overset{(a)}{=} I(M, V; Y|V) - I(M, V; Z|V) \nonumber \\
& \le \max_v \max_{P_{B, X|V=v}} [ I(B; Y|V=v) - I(B; Z|V=v) ] \nonumber \\
  & \le \max_{P_{B, X}}  [ I(B; Y) - I(B ; Z) ],
\end{align}
where
$V_i = (Y^{i-1}, Z_{i+1}^n)$, $P_V(v)=\frac{1}{n}\sum_{i=1}^{n}P_{V_i}(v)$, and $B = (M,V)$.
Here,  (a) follows
from~\cite[Lemma 17.12]{CsiszarKoerner2011InformationTheoryCoding},
and the maximizations are with respect to the constraint
in \cref{eq:esid.converse.singleLetter.divergence},
$D(P_X \circ W_{Z|X}\| Q_Z) \le \frac{\delta^{(n)}}{n}$.

So far, the support of $P_B$ is unbounded.
By~\cite[Lemma~15.4]{CsiszarKoerner2011InformationTheoryCoding},
we can replace $P_B$ by some $P_U \in \cP(\cU)$ such that
$\card\cU \le \card\cX$,
$I(U; Y) - I(U; Z) = I(B; Y) - I(B; Z)$,
and $P_U \circ P_{X|B}(x) = P_B \circ P_{X|B}(x)$,
for every $x \in \cX$:
Similarly to \cite[Lemma~15.5]{CsiszarKoerner2011InformationTheoryCoding},
we arbitrarily select $x_0 \in \cX$ and define the $\card\cX$ many functions
\begin{align}
  f_0(P_X) &= H(Z) - H(Y) = H(P_X \circ W_{Z|X}) - H(P_X \circ W_{Y|X}), \\
  f_x(P)   &= P(x) \text{ for all $x \in \cX \setminus \set{ x_0 }$}.
\end{align}
Thus, by~\cite[Lemma~15.4]{CsiszarKoerner2011InformationTheoryCoding},
there exists a subset $\cU \subseteq \cB$
where $\card\cU \le \card\cX$, and a PMF $P_U \in \cP(\cU)$, that satisfies
\begin{align}
  H(Z|B) - H(Y|B)
     &= \sum_b P_B(b) f_0(P_{X|B=b})
  \\ &= \sum_{u \in \cU} P_U(u) f_0(P_{X|B=u})
  \\ &= H(Z|U) - H(Y|U),
\shortintertext{and for all $x \in \cX \setminus \set{ x_0 }$,}
  P_B \circ P_{X|B}(x)
     &= \sum_b P_B(b) f_x(P_{X|B=b})
  \\ &= \sum_{u \in \cU} P_U(u) f_x(P_{X|B=u})
  \\&= P_U \circ P_{X|B}(x),
\shortintertext{and thus as well}
  P_B \circ P_{X|B}(x_0)
     &= 1 - \sum_{x \ne x_0} P_B \circ P_{X|B}(x)
  \\ &= 1 - \sum_{x \ne x_0} P_U \circ P_{X|B}(x)
  \\ &= P_U \circ P_{X|B}(x_0),
  \\
  I(B; Y) - I(B; Z)
     &= H(P_B \circ P_{X|B} \circ W_{Y|X}) - H(P_B \circ P_{X|B} \circ W_{Z|X})
  \nonumber\\ &\quad
        + H(Z|B) - H(Y|B)
  \\ &= H(P_U \circ P_{X|B} \circ W_{Y|X}) - H(P_U \circ P_{X|B} \circ W_{Z|X})
  \nonumber\\ &\quad
        + H(Z|U) - H(Y|U)
  \\ &= I(U; Y) - I(U; Z)
  .
\end{align}

Since the mutual information is continuous and
the set $\set{ P_{U, X} : D(P_X \circ W_{Z|X} \| Q_Z) \le \delta^{(n)} }$
is compact, we have for $\eta^{(n)} = e^{-\sqrt{n}}$ that
$\alpha^{(n)} = \lambda^{(n)}_1 + \lambda^{(n)}_2 + 2\eta^{(n)} \in o(1)$,
i.e., $\lim\limits_{n \to \infty} \alpha^{(n)} = 0$,
and thus
\begin{align}
  &\capESID(W_{Y|X}, W_{Z|X}, Q_Z) \nonumber
  \\&
    \le
        \lim_{n \to \infty}
        \Biggl[
      o(1)
      +
      \frac{1}{1 - o(1)}
  \twocol{\nonumber\\\ampersand\hspace{7em}}
      \max_{\substack{
        P_{U, X} \in \cP(\cU \times \cX)
        \\ D(P_X \circ W_{Z|X} \| Q_Z) \le \delta^{(n)}/n
        \\ I(U; Y) \ge I(U; Z)
      }}
      I(X; Y)
    \Biggr]
      \\&
    = \max_{\substack{
        P_{U, X} \in \cP(\cU \times \cX)
        \\ P_X \circ W_{Z|X} = Q_Z
        \\ I(U; Y) \ge I(U; Z)
      }}
      I(X; Y).
\end{align}
This completes the proof of \cref{thm:esid.dm}.
\qed

\section{Summary and Discussion}
\label{sec:conclusion}

We introduced the problem of simultaneous approximation of output statistics
over one channel and transmission over another one, where the channel share the
same input (AOST, see \cref{AOSTcodeDef})
and determined the AOST-capacity of a discrete memoryless wiretap channel (DM-WTC)
in \cref{AOST.Theorem}. Further, we introduced the problem of identification
with effective secrecy (ESID) over a DM-WTC,
where a legitimate receiver should perform a reliable identification,
while a warden should not be able to distinguish his output distribution
from a reference, i.e., detect if an unexpected sensible identification
is ongoing, and neither should the warden be able to find out anything
about the content of the signal, i.e., the message.
In contrast to transmission with average-error criteria and
strong stealth or secrecy, where neither of stealth and secrecy implies the other
\cite{Hou14thesis,HouKramer14a,HouKramerBloch2017EffectiveSecrecyReliability},
we considered a maximal-error criterion and naturally semantic secrecy
and stealth, i.e., the worst-case over the messages,
where we noticed that stealth does imply secrecy. Thus,
semantic effective secrecy and semantic stealth are equivalent
(see \cref{seq:securityMeasures.stealth}).

We established upper and lower bounds on the ESID-capacity
in \cref{eq:esid.UpperBound}, \cref{eq:ESID.LowerBound.U}, and \cref{eq.ESID.LowerBound}.
The bounds coincide for more capable channels (\cref{corollary:esid.dm.moreCapable}), but this is a sufficient,
not a necessary condition.
In \cref{sec:example}, we discussed as examples several more capable channels,
but also a pair of reversely degraded channels, which is not more capable.
Hence, for some parameters, there is a gap between the capacity bounds;
yet, for certain other parameters, the bounds coincide.

To prove our ESID results,
we constructed a code based on almost universal hashing with a
random seed, as in
\cite{AhlswedeDueck1989Identificationpresencefeedback,AhlswedeZhang1995Newdirectionstheory},
where the seed is transmitted with a stealthy code, and the hash
with an AOST code. This code achieves our lower bounds.
The upper bounds are proved using a modification of
Watanabe's minimax converse for identification based
on partial resolvability~\cite{Watanabe2022MinimaxConverseIdentification},
combined with techniques
from~\cite{Hou14thesis,HouKramer14a,HouKramerBloch2017EffectiveSecrecyReliability}
and the converse of \cite{AhlswedeZhang1995Newdirectionstheory}.

It \jrwarning*{this is our opinion; we don't know if it is true}{appears}
likely that the lower bound in \cref{corollary:achiev.U} is tight,
based on results from the AOS problem~\cite{HouKramer2013divergence, WatanabeHayashi14},
since the entire codeword is subject to the stealth constraint,
and thus, if an auxiliary pre-channel is used to establish stealth,
intuitively, the resulting virtual channel should be used for identification.
The challenging part of establishing a more stringent converse bound seems to
be the introduction of an auxiliary channel as in
\cref{ReverseDegExamp}, where the number of possible input sequences is suitably
bounded
(see \cref{lemma:idConverse.oneShot} and the discussion of the gap in Ahlswede's
broadcast converse in \cite[Section~2.4]{Bracher2016IdentificationZeroError_phd}).
This is a non-trivial task for the ESID problem and all identification problems
involving broadcast channels.
To the best of our knowledge, there are no existing ID capacity results involving auxiliary variables in the rate bound.
Typically, in converse proofs \cite{ElGamalKim2011NetworkInformationTheory},
the message is obtained as an auxiliary variable from Fano's inequality and then single-letterized.
However, the mutual information between the message and the channel
output cannot serve as an upper bound to the ID capacity,
since ID codes primarily transmit randomness, and not the message.
The capacity of memoryless channels is achieved with codes where
only a few ($\sqrt{n}$) codeword symbols depend on the
message \cite{AhlswedeDueck1989Identificationpresencefeedback}.
On the other hand, in identification converses, one usually approximates
the channel output distribution of the original code using a different code
with bounded complexity, and then the number of messages is bounded by the
complexity of the code. The proof of \cref{thm:esid.dm} follows the same
line of argument. However, this does not reveal any properties
of the code related to the other channel, in a broadcasting or wire-tap setting.

To close the achievability gap, new methods need to be developed to introduce
auxiliary variables in ID converse bounds. This would be a crucial step in the
development of further multi-user converses for ID and many other communication
tasks \cite{Ahlswede2008Generaltheoryinformation,
Ahlswede2021IdenticationOtherProbabilistic},
such as ID via the broadcast channel~\cite{Ahlswede2008Generaltheoryinformation,Bracher2016IdentificationZeroError_phd}.

To this end, observe that any auxiliary channel forms a polytope with extremal
points $P_{X|U=u}$, $u \in \mathcal{U}$, where $\mathcal{U}$ is the auxiliary
alphabet. This polytope must include the set of stealthy encoding distributions.
For such an auxiliary channel concatenated with the communication channel,
the usual ID converse can be applied to obtain an upper bound that matches
the bound in \cref{corollary:achiev.U}, except the cardinality
bound $\card\cU \le \card\cX$ on the auxiliary alphabet.
The latter is also satisfied if the number of
polytope vertices, $\card\cU$, is bounded from above such that
$\log\log \card\cU \in o(n)$.

\ifblind\else

\section*{Acknowledgment}

The authors thank Diego Lentner and Constantin Runge (Technical University of Munich)
for helpful discussions.
Abdalla Ibrahim, Johannes Rosenberger and Christian Deppe
acknowledge the financial support by the Federal Ministry of
Education and Research of Germany (BMBF) in the program of “Souverän. Digital.
Vernetzt.” Joint project 6G-life, project identification number: 16KISK002 and
16KISK263.
Johannes Rosenberger received additional support from the BMBF in the programs
QD-CamNetz (Grant 16KISQ077), QuaPhySI (Grant 16KIS1598K), and QUIET (Grant 16KISQ093).
Abdalla Ibrahim,  Christian Deppe and Roberto Ferrara were further supported in
part by the BMBF within the grant 16KIS1005. Christian Deppe was also supported by the DFG
within the project DE1915/2-1.
Further, Christian Deppe and Roberto Ferrara were supported by the 6GQT project, funded by the
Bavarian Ministry of Economic Affairs, Regional Development, and Energy.
The work of Boulat A.~Bash was supported by the National Science Foundation (NSF) under grants CCF-2006679 and CCF-2045530.
Uzi Pereg was supported by Israel Science Foundation (ISF) under grant 939/23,
in part by German–Israeli Project Cooperation (DIP) within the Deutsche
Forschungsgemeinschaft (DFG) under grant 2032991, in part by the
Ollendorff-Minerva Center (OMC) of Technion under grant 86160946, in part by the
Junior Faculty Program for Quantum Science and Technology of Israel Planning and
Budgeting Committee of the Council for Higher Education (VATAT) under grant
86636903, and in part by the Chaya Career Advancement Chair under grant 8776026.

\fi

\appendix[Proof of \cref{eq:AOSbound}]
\label{appendix:AOSbound}

To derive \cref{eq:AOSbound}, we denote $\cM \coloneqq \cM_{\text{AOST}}$ by $C_m$ the codebook $C_o$ where
the codeword $X(m)$ at position $m$ has been removed.
Therefore,
\begin{align}
  &\expect_{P_{C_o}} \brack{D( P_{Z^n|C_{o}}\| Q_Z^n )}
    = \expect_{P_{C_o}} \expect_{P_{Z^n|C_{o}}}
      \brack{\log\frac{P_{Z^n|C_o}(Z^n|C_o)}{Q_Z^n(Z^n)}} \nonumber\\
  &= \expect_{P_{C_o}}\expect_{P_{M|C_o}}\expect_{P_{Z^n|M,C_o}} \brack{\log\frac{P_{Z^n|C_o}(Z^n|C_o)}{Q_Z^n(Z^n)}} \nonumber \\
  &\overset{(a)}{=}\expect_{P_M}\expect_{P_{C_o}}\expect_{P_{Z^n|M,C_o}}\brack{\log\frac{P_{Z^n|C_o}(Z^n|C_o)}{Q_Z^n(Z^n)}} \nonumber \\
  &\overset{(b)}{=}\expect_{P_M}\expect_{P_{X^n(M)|M}}\expect_{P_{C_M|M}}\expect_{P_{Z^n|M,C_o}}\brack{\log\frac{\sum_{m=1}^{|\mathcal{M}|}
                   \frac 1 {\card\cM} W_{Z|X}(Z^n|X^n(m))}{Q_Z^n(Z^n)}} \nonumber \\
  &\overset{(c)}{=}\expect_{P_M}\expect_{P_{X^n(M)|M}}\expect_{P_{Z^n|X^n(M)}}\expect_{P_{C_M|M}}\Bigg[\log\Big(\frac{W_{Z|X}(Z^n|X^n(M))}{|\mathcal{M}|\, Q_Z^n(Z^n)} \nonumber \\
  &\quad + \frac{\sum_{m=1,m\neq M}^{|\mathcal{M}|}W_{Z|X}(Z^n|X^n(m))}{|\mathcal{M}|\, Q_Z^n(Z^n)}\Big)\Bigg] \nonumber \\
  &\overset{(d)}{\leq} \expect_{P_M}\expect_{P_{X^n(M)|M}}\expect_{P_{Z^n|X^n(M)}}\Bigg[\log\Big(\expect_{P_{C_M|M}}\Big[\frac{W_{Z|X}(Z^n|X^n(M))}{|\mathcal{M}|\, Q_Z^n(Z^n)} \nonumber \\
  &\quad + \frac{\sum_{m=1,m\neq M}^{|\mathcal{M}|}W_{Z|X}(Z^n|X^n(m))}{|\mathcal{M}|\, Q_Z^n(Z^n)}\Big]\Big)\Bigg] \nonumber \\
  &\overset{(e)}{=}\expect_{P_M}\expect_{P_{X^n(M)|M}}\expect_{P_{Z^n|X^n(M)}}\Bigg[\log\Big(\frac{W^n_{Z|X}(Z^n|X^n(M))}{|\mathcal{M}|\, Q_Z^n(Z^n)} \nonumber \\
  &\quad + \frac{\sum_{m=1, m\neq M}^{|\mathcal{M}|}\sum_{x^n \in \cX^n}P^n_X(x^n)W_{Z|X}^{n}(Z^n|x^n)}{|\mathcal{M}|\, Q_Z^n(Z^n)}\Big)\Bigg] \nonumber \\
  &\overset{(f)}{=}\expect_{P_M}\expect_{P_{X^n(M)}}\expect_{P_{Z^n|X^n(M)}}\Bigg[\log\Big(\frac{W_{Z|X}^{n}(Z^n|X^n(M))}{|\mathcal{M}|\, Q_Z^n(Z^n)} +\frac{|\mathcal{M}|-1}{|\mathcal{M}|}\Big)\Bigg] \nonumber \\
  & \leq\expect_{P_M}\expect_{P_{X^n(M)}}\expect_{P_{Z^n|X^n(M)}}\Bigg[\log\Big(\frac{W_{Z|X}^{n}(Z^n|X^n(M))}{|\mathcal{M}|\, Q_Z^n(Z^n)}+ 1 \Big)\Bigg],
  \label{AchievabilityStealth2}
\end{align}
where (a) holds as $C_{o}$ and $M$ are statistically independent;
(b) holds as $P_{C_o}= P_{X^n(M)} P_{C_M}$ where $C_M$ and $X^n(M)$ are statistically independent;
(c) holds as $ C_M - X^n(M) - Z^n$ ;
(d) holds due to Jensen's inequality $\expect(f(X)) \leq f(\expect(X))$ for a concave $f(X)=\log X$;
while (e) holds from the linearity of the expectation
  and the statistical independence of the codewords
  $X^n(m) \forall m \in \mathcal{M}$, and hence
\begin{equation}
\expect_{P_{C_M}} \brack{ \sum_{m=1,m\neq M}^{|\mathcal{M}|} f(X^n(m), Y^n) }
  = \sum_{m=1,m\neq M}^{|\mathcal{M}|} \expect_{P_{X^n(m)}} \brack{ f(X^n(m), Y^n) }
  .
\end{equation}
Further, (f) follows since for all $P_X \in \Paost$,
\begin{equation}
  Q_Z(z_i) = \sum_{x \in \cX} P_X(x) W_{Z|X}(z_i | x)
\end{equation}
for all \( m \in \mathcal{M} \) and \( i \in [1:N] \).
\qed

\clearpage
\printbibliography

\newpage
\listoffixmes

\end{document}

